\documentclass{article}
\usepackage{arxiv}
\usepackage{hyperref}
\usepackage[utf8]{inputenc} 
\usepackage[T1]{fontenc}    
\usepackage{amsfonts} 
\usepackage{url}           
\usepackage{booktabs}       
\usepackage{nicefrac}      
\usepackage{microtype}      

\usepackage{amsthm}
\usepackage{amsmath}
\usepackage{amssymb}
\usepackage{mathtools}
\usepackage{scalefnt}
\usepackage{arydshln}
\usepackage{nth}
\usepackage{multirow}
\newtheorem{theorem}{Theorem}[section]

\newtheorem{lemma}[theorem]{Lemma}
\newcommand{\mtimes}{{\mkern-2mu\times\mkern-2mu}}

\title{Symmetry and the salience of textures}

\author{
  Marconi Barbosa\thanks{\texttt{marconi.barbosa@anu.edu.au}}\hfill Ted Maddess\\
  Eccles Institute of Neuroscience\\
  The John Curtin School of Medical Research\\
  The Australian National University\\
  }

\begin{document}
\maketitle

\begin{abstract}
In this paper we investigate the role of symmetry in visual stimuli designed to probe human sensitivity to image statistics. Our starting point is a recently published parameter space, a point in which defines a family of binary texture images displaying a prescribed content of \nth{1}- to \nth{4}-order correlations among pixels in $2\mtimes2$ patches. We show that this parameter space can be represented by fewer variables, namely the \emph{orbit invariants} obtained by exploiting texture symmetry. Next we show how a class of locally countable texture statistics, the Minkowski functionals -- recently shown to be a proxy for human performance in texture discrimination tasks - can be written as a linear combination of the dihedral orbit invariants. Furthermore, by recasting these functionals as a combination of dihedral invariants, a generalization of these functionals can be obtained for textures of any number of grey-levels, patch sizes, or lattice types -- greatly reducing the number of dimensions/parameters needed to characterize the generated images. Orbit invariants may therefore provide a clue on the discrimination of these richer textures, as the ordinary Minkowski functionals do for binary textures. 
\end{abstract} 


\section*{Introduction}\label{sec:int}
The multisensory percept of {\it texture} is, in common language, frequently associated with the tactile or visual properties of surfaces and its finer spatial details~\cite{klatzkyMultisensoryTexturePerception2010}. Although many of the properties addressed in this paper would be shared by other senses, we are concerned here with textures as visual stimuli.  

Sensitivity to variations in texture is a prime aspect of early vision, if we consider for a moment the precarious optics and low visual information content arising from the environment. A clear advantage for creatures living in murky, turbid waters who can discern what is just a familiar level (and type) of noise -- so it can maintain its food searching routine, from a noise suggesting an inconspicuous threat may be lurking among the pixels. It turns out that visual texture perception or structural vision seems to create an evolutionary pressure, leading at least to some odd specializations~\cite{kreysingPhotonicCrystalLight2012}. 

In this study, we refer to a texture as a visual stimuli, algorithmically generated family of bidimensional binary images with similar, controlled, {\it level} of disorder. How conspicuous a texture would be if compared to a reference texture? The intensity of this relative percept is what we call here {\it salience}. Perhaps surprisingly, when the reference texture is generated by enforcing that all possible configurations of the $2\mtimes2$ patches are present uniformly, salience is not quantified by entropy, the trusted mathematical measure of disorder. So the perceptual distance of a texture from pure noise is a multidimensional concept. This suggest the visual system evolved developing neural mechanisms sensitive to structural properties in textures beyond (and perhaps other than) the standard notion of disorder. Despite a large body of research--across several disciplines such as psychophysics, neuroscience, and  computer vision, the spatial statistical properties intrinsic to the texture stimuli which make them conspicuous is not fully understood. 

\begin{figure*}
\centering
\includegraphics[trim = 0cm 0cm 0cm 0cm, clip,width = 0.3\linewidth]{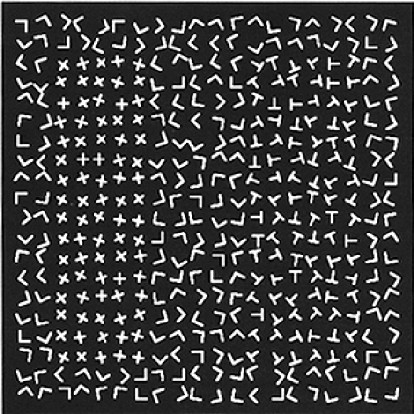}
\vspace{0.2cm}
\caption{Local spatial statistics of a texture are responsible for inducing salience. A) On the left of this image, there is a region containing randomly oriented five pixel correlation motif (X), and on the right a similar region containing a motif based on a four pixel (T) correlation. Relative to the background containing randomly rotated three pixel correlation motif (L) the region with five pixel correlation is more salient.\label{fig:saliency}}
\end{figure*}
 
It is known however that, for some synthetic and natural texture families, salience is associated with the proportion of third and higher order pixel correlations content in the texture family~\cite{tkacikLocalStatisticsNatural2010,juleszTextonsElementsTexture1981,caelliPerceptualAnalyzersUnderlying1978a}. 

We have reported on bounds for the number of mechanisms explaining human texture discrimination performance~\cite{seamonsLowerBoundNumber2015}. That result is supported by cross-validated outcomes of a large crowd-sourced psychophysics experiment of 121 participants~\cite{seamonsUnderlyingNeuralMechanisms2016}, and work on a set of texture families displaying 10 controlled image statistics up to 4th-order~\cite{victorPerceptualSpaceLocal2015}. Collectively these studies suggest 7 or fewer mechanisms operate to impart our sense of pixel correlations up to 4th-order. What aspects of texture content might these correspond to? We have provided evidence~\cite{barbosaLocallyCountableProperties2013,seamonsUnderlyingNeuralMechanisms2016}, that the strongest indicator of the salience for 33 different texture families are moments of the distributions of a special class of image functionals, namely the Minkowski functionals -- particularly of the Euler number ($\chi$). As we will show in this paper, these image functionals arise naturally in the characterization of symmetry in binary textures and provide a parsimonious way to define general texture content. Furthermore, we suggest alternative measures for more natural textures for which the Minkowski functionals are not well defined, broadening the role these image functionals might have in texture discrimination.  

A method originally described in~\cite{victorLocalImageStatistics2012,victorPerceptionSecondThirdorder2013,victorPerceptualSpaceLocal2015} produces textures with set pixel correlations of a particular order, while all other pixel correlations are as random as possible. This allows dissection of the human sensitivity to all orders and types of correlations defined for $2\mtimes2$ pixel patches. The probability of occurrence of these     
$2\mtimes2$ patches seems a natural set of parameters to control in the design of textures with the prescribed spatial statistical properties. But these parameters over represent the space $V_{n}^{q\mtimes q}$ with $q=2$ of possible binary ($n=2$) texture families. A multidimensional Fourier transform of these probabilities forms a dual space $\hat{V}_2^{2\mtimes 2}$, whose basis, the Fourier coefficients, consist of a smaller subset of texture statistics, once spatial stability constraints and consistency are enforced. We will show in this paper how this number can be reduced further by observing symmetries intrinsic to these textures and even further, as far as the visual system is concerned, if we adopt recent empirical evidence~\cite{barbosaLocallyCountableProperties2013,seamonsLowerBoundNumber2015}.

This paper has four main contributions. Firstly it provides a pathway for dimensionality reduction based on symmetry. Secondly it relates the restricted set of coordinates obtained with previous work showing they are relevant to perceptual salience. Furthermore, this framework works unchanged with any number of grey-levels, neighbourhood sizes and lattice symmetries, providing a natural extension for the Minkowski Functionals (and explanatory models of discrimination) for richer textures with multiple grey-levels. Finally we also illustrate the use of these invariants and of functions defined on orbits in the characterization of all the axial (single controlled coordinate) and planar (mixtures) binary textures generated by the methods in~\cite{victorLocalImageStatistics2012,victorPerceptionSecondThirdorder2013,victorPerceptualSpaceLocal2015}.

\section*{Methods}

\subsection*{Textures with specified spectrum}

A framework for synthesizing textures with a specified structural (correlations of up to 4 pixels, nearest neighbors) content was introduced in~\cite{victorLocalImageStatistics2012,victorPerceptualSpaceLocal2015}. Textures were generated by maximum entropy principle with specific constraints controlling for each type of correlation. An exhausting process for coordinate types (correlations) and their mixtures is described in detail on~\cite{victorPerceptualSpaceLocal2015}. The nomenclature for describing the types correlations and coordinates is given in Figure~\ref{tb:quads}. We use the same terminology and symbols for continuity. Central to this mixing procedure is the \emph{doughnut algorithm} which swaps pixels with identical neighbourhoods preserving short range correlations but disrupting long range and higher-order correlations. See Figure~\ref{fig:axial} for examples of how the appearance of these axial textures change when a single coordinate changes value. Likewise, see Figure~\ref{fig:planar}for composite, planar textures, resulting from a modulation of two coordinates. We will study these compositions later on in terms (to be defined) of symmetry: their dihedral (orbital) entropy and invariance.

\begin{figure*}
\centering
\includegraphics[trim = 0cm 0cm 0cm 0cm, clip,width = 0.5\linewidth]{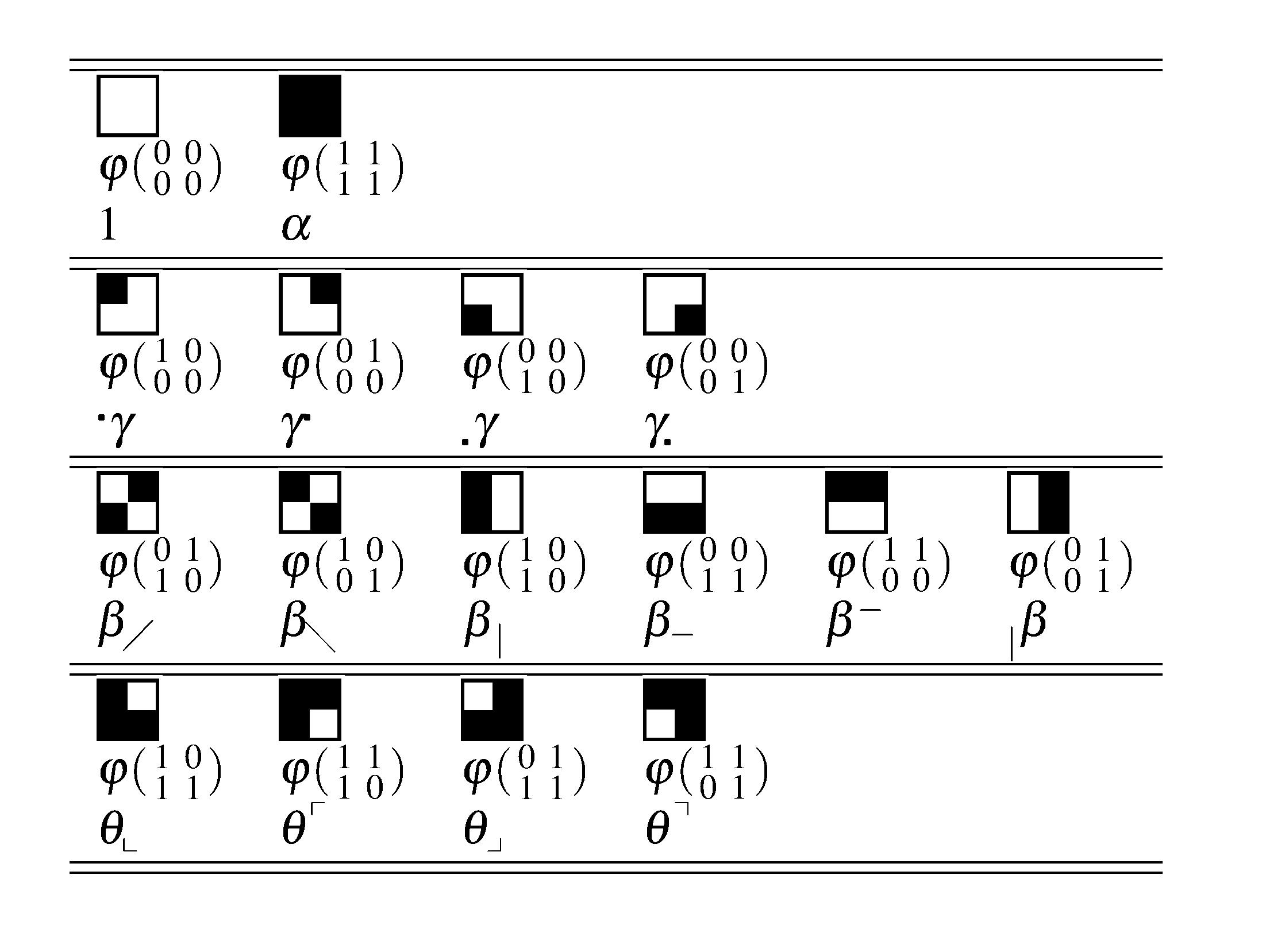}
\caption{Coordinates of $V_2^{2\mtimes2}$ representing different orders (see~\cite{victorLocalImageStatistics2012}) of correlation: $\gamma, $\nth{1}; $\beta$, \nth{2}; $\theta$, \nth{3}; $\alpha$, \nth{4}.\label{tb:quads}}
\end{figure*}

\begin{figure}[htbp]
\centering
\includegraphics[trim = 0cm 0cm 0cm 1cm, clip,width = 0.85\linewidth]{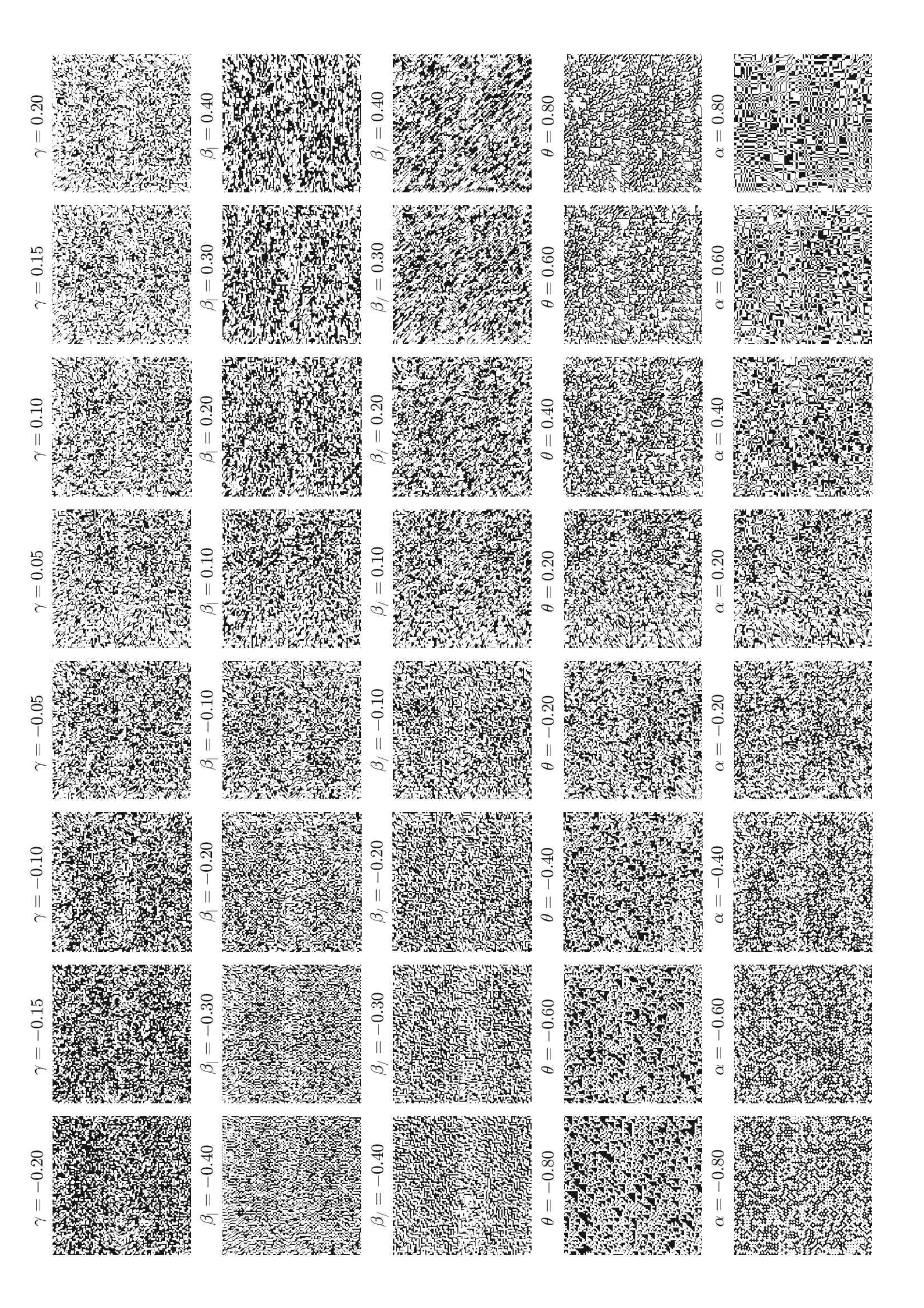}
\caption{Samples of single coordinate (axial) textures representing directions in the block pixel constrained texture space. A value of 0 represents the origin (binary noise) while 1 represent pure textures with no noise.~\label{fig:axial}}
\end{figure}

\begin{figure}[htbp]
\centering
\includegraphics[trim = 0cm 0cm 1cm 1cm, clip,width = 0.88\linewidth]{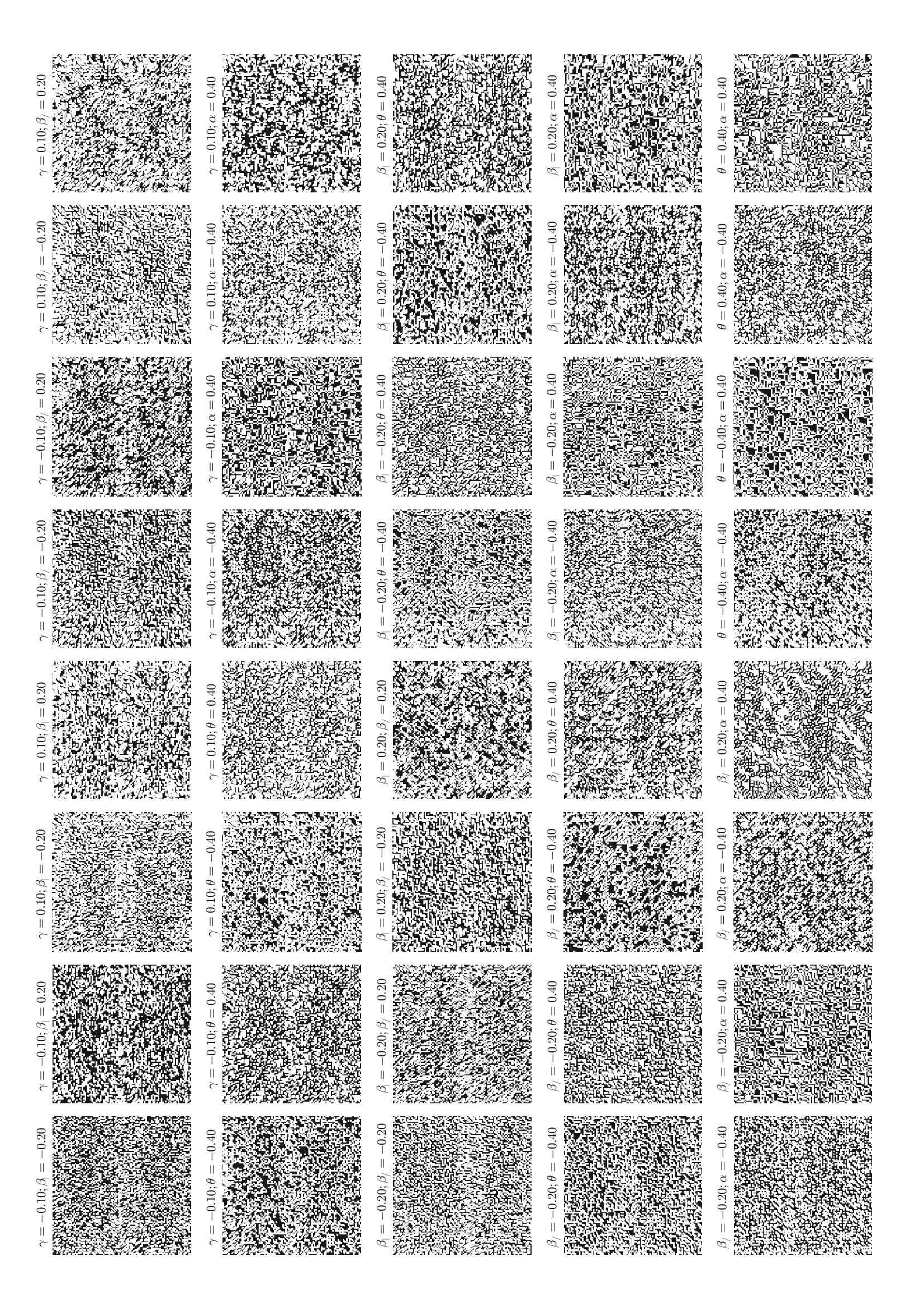}
\caption{Samples of planar textures. These textures are formed by mixing two axial directions, defining a plane in $V_2^{2\mtimes2}$, and hence are referred to as planar textures.\label{fig:planar}}
\end{figure}

\subsection*{Salience descriptors}
We are ultimately interested in features of textures which could be calculated for many grey-levels and larger neighbourhoods. We want these features to also have predictive power to explain the human discrimination performance of these textures. Using all coordinates of the perception space is prohibitive as their cardinality grows steeply with the number of grey-levels and pixel patch size, see Table~\ref{tb:coords-orbits}. For binary textures, it has been shown that moments of the Minkowski functionals are a good proxy of human texture discrimination performance~\cite{barbosaLocallyCountableProperties2013}, but its definition in terms of neighbouring pixel correlations or probability of pixel patches is confined to binary form.
Because of the fuzzy nature of connectivity in grey-level images, various proposed generalizations seem equally justifiable~\cite{soilleGenuineConnectivityRelations2007,braga-netoGrayscaleLevelConnectivity2004,heijmansTheoreticalAspectsGraylevel1991,wangDigitalConnectivityExtended1997}. With the explicit relationship between the Minkowski functionals and orbit invariants, we can latch on to a concept that is well defined for any number of grey-levels and patch size, expecting them to also be predictive of Human perception in future experiments with grey-level textures. 

\subsection*{Orbits}
We follow a slimlined approach to be able to state the main results without much technicalities. Details of relevant group theoretic formalism is postponed to a support information section in the end. But here we only state the definition of the orbit of an object under the action of a group $G$,  $G(x) = Orb(x) = \{y\in X| y = g*x\}$.  That is, $G(x)$ is the set of all possible destinations of object $x$ under the group action.

\begin{table}[htbp]
\centering
\begin{tabular}{l|l|l}\hline
$n$\textbackslash  $N$ 	& $2$  		& $3$ \\ \hline\hline
2            		  	& 16-6    	&  512-102\\ \hline
3            			& 81-21    	&  19683-2862\\ \hline
4            			& 256-55    &  262144-? \\\hline\hline
\end{tabular}
\vspace{0.5cm}
\caption{Number of coordinates (first number) and of pure dihedral orbits, for a combination of grey-levels($n$) and size of the neighbourhood ($N$).\label{tb:coords-orbits}}
\vspace{0.5cm}
\end{table}

Table~\ref{tb:coords-orbits} shows how the number of coordinates and orbits scale with more grey-levels and larger neighbourhoods. Table~\ref{tb:orbits} shows the action of the group $D_4^2 = D_4\times K^2$ on the basis of the perception space $V_2^{2\mtimes2}$, $\varphi(\begin{smallmatrix} s_1 & s_2\\ s_3 & s_4 \end{smallmatrix})$. The first column in the left shows the index for each of the coordinates of $V_2^{2\mtimes2}$, they are displayed in a different order from how they appears in Equations~\ref{eq:jvtransform}, to make orbits explicit. We show on the right side of the table the cycle notation for the alphabetically labelled elements of $D_4^2$. The last row shows the cardinality of the fixed space for each group element -- the sum of the number of fixed coordinates, which are shown in bold face. Using Theorem~\ref{lem:Burns} we obtain by summing up the last row and dividing by the corresponding group order, the $6$ orbits for $D_4$ group (including dashed lines divisions) which merges into $4$ orbits (separated by solid lines only) when contrast reversal is included (group $D_4^2$).

\begin{table}[htbp]
\vspace{0.5cm}
\parbox{.7\linewidth}{
\scalebox{0.7}{
\centering
{\fontfamily{cmtt} \selectfont \small
\begin{tabular}{|ll|llllllllllllllllll|}
\hline
\multicolumn{2}{|c|}{Orbit} & ind & quad & a & b & c & d & e & f & g & h & a* & b* & c* & d* & e* & f* & g* & h* \\ \hline
\multirow{ 2}{*}{$\mathcal{O}_0^1$} & $\mathcal{O}_0^0 $ & $1$ & \scalebox{0.8}{$\begin{bsmallmatrix}
0 & 0 \\
0 & 0
\end{bsmallmatrix}$} & $\mathbf{1}$ & $\mathbf{1}$ & $\mathbf{1}$ & $\mathbf{1}$ & $\mathbf{1}$ & $\mathbf{1}$ & $\mathbf{1}$ & $\mathbf{1}$ & $16$ & $16$ & $16$ & $16$ & $16$ & $16$ & $16$ & $16$ \\\cdashline{2-20}
& $\mathcal{O}_5^0$ & $16$ & \scalebox{0.8}{$\begin{bsmallmatrix}
1 & 1 \\
1 & 1
\end{bsmallmatrix}$} & $\mathbf{16}$ & $\mathbf{16}$ & $\mathbf{16}$ & $\mathbf{16}$ & $\mathbf{16}$ & $\mathbf{16}$ & $\mathbf{16}$ & $\mathbf{16}$ & $1$ & $1$ & $1$ & $1$ & $1$ & $1$ & $1$ & $1$ \\ \hline
\multirow{ 8}{*}{$\mathcal{O}_1^1$} & \multirow{ 4}{*}{$\mathcal{O}_1^0$} & $2$ & \scalebox{0.8}{$\begin{bsmallmatrix}
1 & 0 \\
0 & 0
\end{bsmallmatrix}$} & $\mathbf{2}$ & $5$ & $3$ & $4$ & $5$ & $3$ & $4$ & $\mathbf{2}$ & $12$ & $13$ & $14$ & $15$ & $13$ & $14$ & $15$ & $12$ \\
& & $3$ & \scalebox{0.8}{$\begin{bsmallmatrix}
0 & 1 \\
0 & 0
\end{bsmallmatrix}$} & $\mathbf{3}$ & $\mathbf{3}$ & $5$ & $5$ & $4$ & $2$ & $2$ & $4$ & $14$ & $14$ & $13$ & $13$ & $15$ & $12$ & $12$ & $15$ \\
& & $4$ & \scalebox{0.8}{$\begin{bsmallmatrix}
0 & 0 \\
1 & 0
\end{bsmallmatrix}$} & $\mathbf{4}$ & $\mathbf{4}$ & $2$ & $2$ & $3$ & $5$ & $5$ & $3$ & $15$ & $15$ & $12$ & $12$ & $14$ & $13$ & $13$ & $14$ \\
& & $5$ & \scalebox{0.8}{$\begin{bsmallmatrix}
0 & 0 \\
0 & 1
\end{bsmallmatrix}$} & $\mathbf{5}$ & $2$ & $4$ & $3$ & $2$ & $4$ & $3$ & $\mathbf{5}$ & $13$ & $12$ & $15$ & $14$ & $12$ & $15$ & $14$ & $13$ \\ \cdashline{2-20}
& \multirow{ 4}{*}{$\mathcal{O}_4^0$} & $12$ & \scalebox{0.8}{$\begin{bsmallmatrix}
0 & 1 \\
1 & 1
\end{bsmallmatrix}$} & $\mathbf{12}$ & $13$ & $14$ & $15$ & $13$ & $14$ & $15$ & $\mathbf{12}$ & $2$ & $5$ & $3$ & $4$ & $5$ & $3$ & $4$ & $2$ \\
& & $13$ & \scalebox{0.8}{$\begin{bsmallmatrix}
1 & 1 \\
1 & 0
\end{bsmallmatrix}$} & $\mathbf{13}$ & $12$ & $15$ & $14$ & $12$ & $15$ & $14$ & $\mathbf{13}$ & $5$ & $2$ & $4$ & $3$ & $2$ & $4$ & $3$ & $5$ \\
& & $14$ & \scalebox{0.8}{$\begin{bsmallmatrix}
1 & 0 \\
1 & 1
\end{bsmallmatrix}$} & $\mathbf{14}$ & $\mathbf{14}$ & $13$ & $13$ & $15$ & $12$ & $12$ & $15$ & $3$ & $3$ & $5$ & $5$ & $4$ & $2$ & $2$ & $4$ \\
& & $15$ & \scalebox{0.8}{$\begin{bsmallmatrix}
1 & 1 \\
0 & 1
\end{bsmallmatrix}$} & $\mathbf{15}$ & $\mathbf{15}$ & $12$ & $12$ & $14$ & $13$ & $13$ & $14$ & $4$ & $4$ & $2$ & $2$ & $3$ & $5$ & $5$ & $3$ \\ \hline
\multirow{ 4}{*}{$\mathcal{O}_2^1$} & \multirow{ 4}{*}{$\mathcal{O}_2^0$} & $6$ & \scalebox{0.8}{$\begin{bsmallmatrix}
0 & 0 \\
1 & 1
\end{bsmallmatrix}$} & $\mathbf{6}$ & $9$ & $9$ & $8$ & $8$ & $\mathbf{6}$ & $11$ & $11$ & $8$ & $11$ & $11$ & $\mathbf{6}$ & $\mathbf{6}$ & $8$ & $9$ & $9$ \\
& & $8$ & \scalebox{0.8}{$\begin{bsmallmatrix}
1 & 1 \\
0 & 0
\end{bsmallmatrix}$} & $\mathbf{8}$ & $11$ & $11$ & $6$ & $6$ & $\mathbf{8}$ & $9$ & $9$ & $6$ & $9$ & $9$ & $\mathbf{8}$ & $\mathbf{8}$ & $6$ & $11$ & $11$ \\
& & $9$ & \scalebox{0.8}{$\begin{bsmallmatrix}
1 & 0 \\
1 & 0
\end{bsmallmatrix}$} & $\mathbf{9}$ & $6$ & $8$ & $\mathbf{9}$ & $11$ & $11$ & $6$ & $8$ & $11$ & $8$ & $6$ & $11$ & $\mathbf{9}$ & $\mathbf{9}$ & $8$ & $6$ \\
& & $11$ & \scalebox{0.8}{$\begin{bsmallmatrix}
0 & 1 \\
0 & 1
\end{bsmallmatrix}$} & $\mathbf{11}$ & $8$ & $6$ & $\mathbf{11}$ & $9$ & $9$ & $8$ & $6$ & $9$ & $6$ & $8$ & $9$ & $\mathbf{11}$ & $\mathbf{11}$ & $6$ & $8$ \\ \hline
\multirow{ 2}{*}{$\mathcal{O}_3^1$} & \multirow{ 2}{*}{$\mathcal{O}_3^0$} & $7$ & \scalebox{0.8}{$\begin{bsmallmatrix}
1 & 0 \\
0 & 1
\end{bsmallmatrix}$} & $\mathbf{7}$ & $\mathbf{7}$ & $10$ & $10$ & $\mathbf{7}$ & $10$ & $10$ & $\mathbf{7}$ & $10$ & $10$ & $\mathbf{7}$ & $\mathbf{7}$ & $10$ & $\mathbf{7}$ & $\mathbf{7}$ & $10$ \\
& & $10$ & \scalebox{0.8}{$\begin{bsmallmatrix}
0 & 1 \\
1 & 0
\end{bsmallmatrix}$} & $\mathbf{10}$ & $\mathbf{10}$ & $7$ & $7$ & $\mathbf{10}$ & $7$ & $7$ & $\mathbf{10}$ & 	$7$ & $7$ & $\mathbf{10}$ & $\mathbf{10}$ & $7$ & $\mathbf{10}$ & $\mathbf{10}$ & $7$ \\\hline\hline
\multicolumn{2}{|c}{$|X^g|$}& & & $16$ & $8$ & $2$ & $4$ & $4$ & $4$ & $2$ & $8$ & $0$ & $0$ & $2$ & $4$ & $4$ & $4$ & $2$ & $0$ \\ \hline
\end{tabular}}}}
\hspace{1.5cm}
\parbox{0.01\linewidth}{
\scalebox{0.8}{
\centering
{\fontfamily{cmtt} \selectfont \small
\begin{tabular}{ll}
a & $(id)$ \\
b & $(1,7)(2,8)$ \\
c & $(1,3,7,5)(2,4,8,6)$ \\
d & $(1,5)(2,6)(3,7)(4,8)$ \\
e & $(1,7)(2,8)(3,5)(4,6)$ \\
f & $(1,3)(2,4)(5,7)(6,8)$ \\
g & $(1,5,7,3)(2,6,8,4)$ \\
h & $(3,5)(4,6)$ \\
a* & $(1,2)(3,4)(5,6)(7,8)$ \\
b* & $(1,8)(2,7)(3,4)(5,6)$ \\
c* & $(1,4,7,6)(2,3,8,5)$ \\
d* & $(1,6)(2,5)(3,8)(4,7)$ \\
e* & $(1,8)(2,7)(3,6)(4,5)$ \\
f* & $(1,4)(2,3)(5,8)(6,7)$ \\
g* & $(1,6,7,4)(2,5,8,3)$ \\
h* & $(1,2)(3,6)(4,5)(7,8)$ \\
\end{tabular}
}}
}\vspace{0.5cm}
\caption{Binary orbits for the action of elements of the dihedral group $D_{4}$ and the product $D_4\mtimes K^2$. Orbits $\mathcal{O}_j^1$ for the group $D_4\mtimes K^2$ are separated by a solid line. The pure $D_4$ orbits, $\mathcal{O}_j^0$, are separated by dashed and solid lines. Elements containing a contrast reversal, the non-trivial element in $K^{2}$, are marked with $*$. \label{tb:orbits}}\vspace{0.3cm}
\end{table}

Table~\ref{tb:coords-orbits} was generated by brute force, explicitly constructing all orbits using a computer. But techniques from group theory allow us to obtain a closed form for the number of orbits. We illustrate here use of the Lemma~\ref{lem:Burns} in order to derive a general formula for $n$ grey-levels in a $2\mtimes2$  neighbourhood.

\begin{theorem}[Orbit Counting]\label{lem:Burns}
The number of orbits of a group $G$ acting on a set $X$ is the average number of fixed points in $X$ by the action of elements of $G$.
\end{theorem}
\begin{proof}
See Supplementary Information.
\end{proof}
We display in Table~\ref{tb:coords-orbits-closed} the number of fixed coordinates by the action of the Dihedral group $D_4$. On the left side of this table the coordinates are divided by their order and types.
\begin{table}[htbp]
\vspace{0.5cm}
\centering
\begin{tabular}{|l|l|l|l|l|l|l|l|l|l|l}\hline
 $x$ \textbackslash $g$
& $e$ 	      & $r$  & $r^{-1}$ & $r^2$      & $h$      & $v$         & $s$        &  $b$  \\ \hline\hline
$\mathbin{^{\centerdot}\gamma},\gamma^{\centerdot},\mathbin{_{\centerdot}\gamma},\gamma_{\centerdot} $
& $4(n-1)$    & $0$  & $0$      & $0$        & $0$      & $0$         & $2(n-1)$   &  $2(n-1)$ 		\\ \hline
$\beta_{-},\beta^{-}$
& $2(n-1)^2$  & $0$  & $0$      & $0$        & $0$      & $2(n-1)$    & $0$        &  $0$    		\\ \hline
$\mathbin{_{}^{}\beta}_{|},\mathbin{_|^{}\beta}_{}$	
& $2(n-1)^2$  & $0$  & $0$      & $0$        & $2(n-1)$ & $0$         & $0$        &  $0$    		\\ \hline
$\beta_{\diagup},\beta_{\diagdown}$
& $2(n-1)^2$  & $0$  & $0$      & $2(n-1)$   & $0$      & $0$         & $n(n-1)$   &  $n(n-1)$     	\\ \hline
$\theta^{\ulcorner},\theta^{\urcorner},\theta_{\llcorner},\theta_{\lrcorner}$
& $4(n-1)^3$  & $0$  & $0$      & $0$        & $0$      & $0$         & $2(n-1)^2$ &  $2(n-1)^2$    \\ \hline
$\alpha$
& $(n-1)^4$   & $(n-1)$  & $(n-1)$      & $(n-1)^2$   & $(n-1)^2$ & $(n-1)^2$ & $(n-1)^3$ &  $(n-1)^3$ \\ \hline
$\varphi\scalebox{0.8}{$\begin{bsmallmatrix}
0 & 0 \\
0 & 0
\end{bsmallmatrix}$}$
& $1$ & $1$ & $1$ & $1$  & $1$ & $1$ & $1$ & $1$ \\ \hline \hline
$|X^g|$ & $n^4$ &  $n$ & $n$& $n^2 $ & $n^2$ &$n^2 $ & $n^3 $& $n^3$
\\ \hline
\end{tabular}
\vspace{0.5cm}
\caption{Number of coordinates, for each type, fixed by the elements of dihedral group $D_4$, namely: the identity $e$, $\frac{\pi}{2}$-rotation $r$, its inverse $r^{-1}$, $\pi$-rotation $r^2$, ($h$)orizontal and ($v$)ertical reflections, ($s$)lash and ($b$)ackslash  diagonal reflections.\label{tb:coords-orbits-closed}}
\end{table}
Summing each column of Table~\ref{tb:coords-orbits-closed}, we arrive at the last row of this table which shows cardinality of the fixed subspaces of each element of $g$ of $G$. Then summing up this row, i.e all fixed subspaces for each element of $G$ and dividing by the order of $G$ we have
\begin{equation}
\frac{1}{|G|}\sum_{g\in G}|X^g| = \frac{n(n^3 + 2n^2 + 3n + 2)}{8}
\end{equation}
By Theorem~\ref{lem:Burns}, this number is also the number of orbits $ |X/G|$ for $n$ grey-levels.
We see from Figure~\ref{fig:scaling} that $|X/G|$ grows significantly slower than the number of coordinates. \emph{This suggests that functions constant over orbits are one way to reduce the number of descriptors in texture discrimination}. Also, the number of orbits decreases even further if contrast inversion is allowed.  We will see supporting evidence of this in the subsequent sections, when we show the relationship between these invariant orbit functions with the Minkowski functionals.
\begin{figure}[htbp]
\centering
\includegraphics[trim = 0cm 0cm 0cm 0cm, clip,width = .75\linewidth]{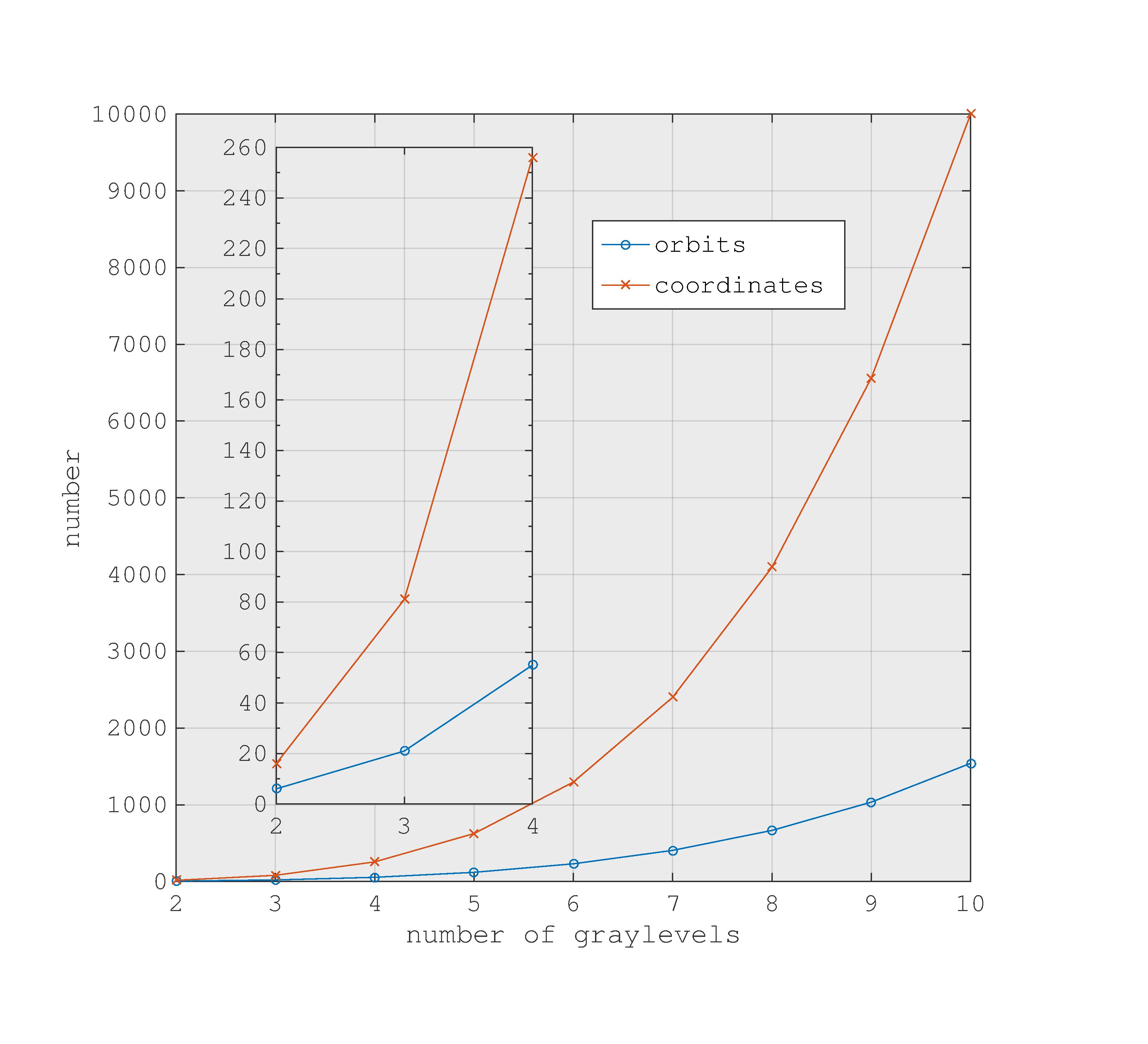}
\caption{Scaling of the number of orbits and coordinates as a function of the number of grey-levels in  textures defined by $2\mtimes2$ pixel patches. The insert shows a zoomed portion of the plot near the origin. \label{fig:scaling}}
\end{figure}

\begin{figure}[htbp]
\centering
\includegraphics[trim = 0cm 0cm 0cm 0cm, clip,width = .75\linewidth]{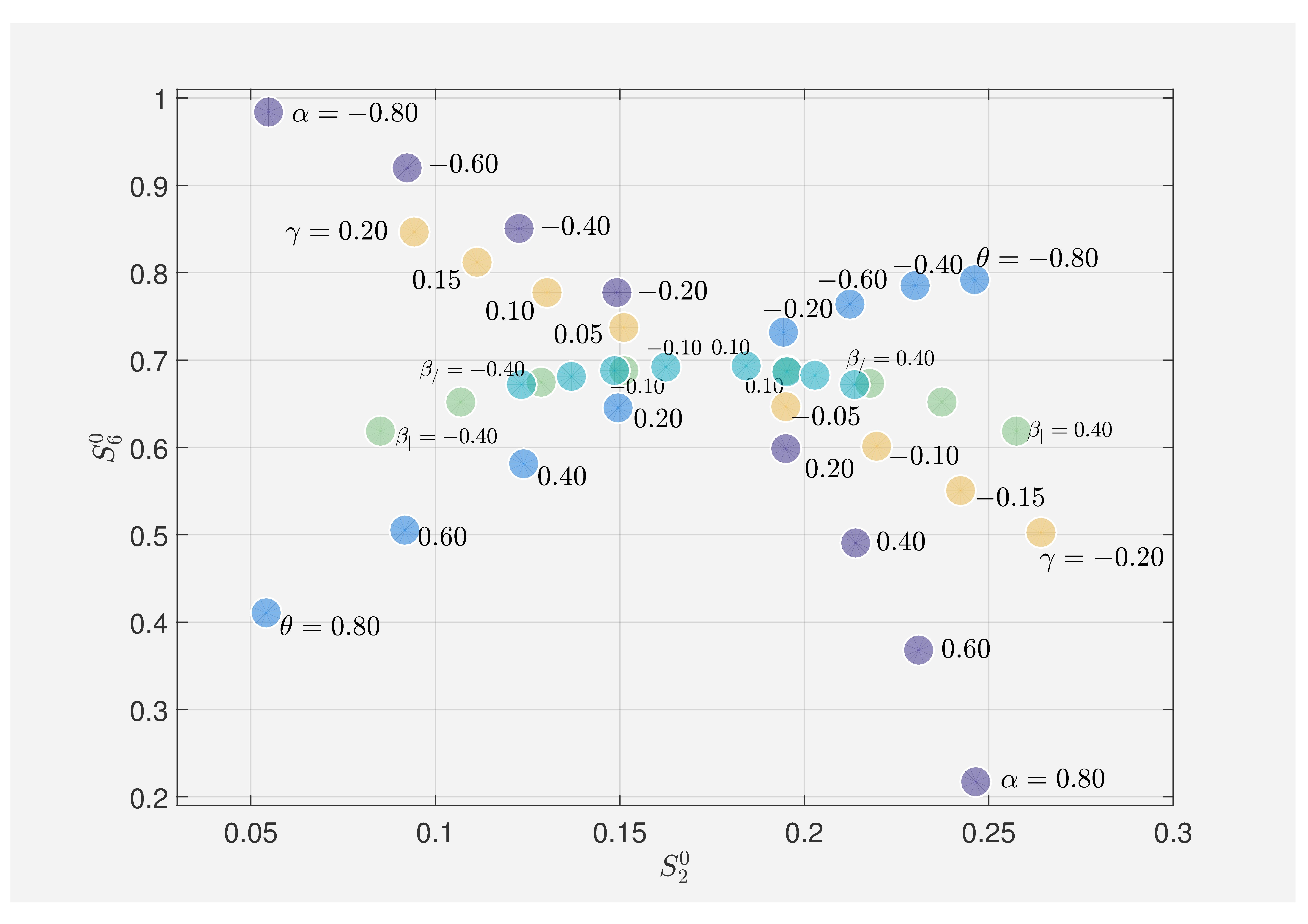}
\caption{Scatter plot of dihedral entropies $S_{2}^{0},S_{6}^{0}$, see definition in Equation~\ref{eq:partent}, for all single coordinate (axial) textures of Figure~\ref{fig:axial}. The superscript $0$ indicates lack of contrast reversal. For clarity, only the first and last steps are labelled for $\beta_{|}$  and $\beta_{/}$; steps shown in $0.1$ increments in the respective coordinate value, from $-0.4$ to $0.4$. The central region in this picture is of higher standard BGSvN entropy.\label{fig:scatteraxial}}
\end{figure}

\begin{figure}[htbp]
\centering
\includegraphics[trim = 0cm 0cm 0cm 0cm, clip,width = .75\linewidth]{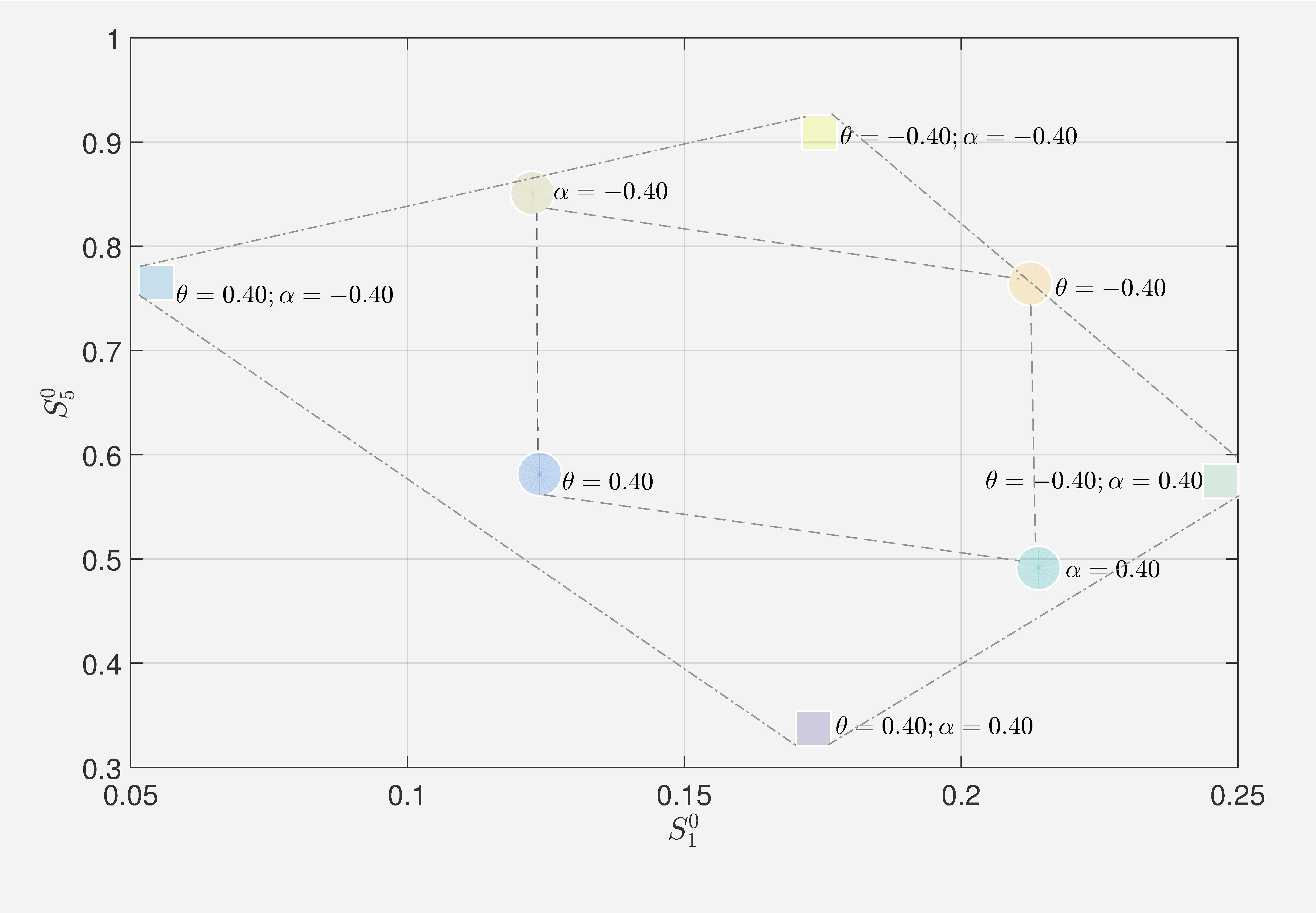}
\caption{Dihedral entropy ($S_1^0, S_5^0$, see definition in Equation~\ref{eq:partent}), for a particular composition of textures. The superscript $0$ indicates lack of contrast reversal. These salience descriptors are calculated for single coordinates textures (examples in Figure~\ref{fig:axial}), displayed here by discs. For the planar textures, compositions of axial textures~ Figure~\ref{fig:planar}, the entropies is displayed here with squares.\label{fig:scattercomb}}
\end{figure}

\subsection*{Orbit Entropy.}
We describe in this section image features related to symmetry which are descriptive of texture structure. The Boltzmann-Gibbs-Shannon-von Neumann (BGSvN) entropy $S=-\sum_i p_i\ln p_i$ is the only functional of pixel patch probability $p_i$ whose maximization leads to a bias free estimation of the underlying probability density $p$, in additive systems. Forays in generalizing the concept of entropy for non-additive systems, usually with application to highly correlated (long range) complex systems away from equilibrium, have been proposed by various researchers, but not without controversy~\cite{presseNonadditiveEntropiesYield2013,presseNonadditiveEntropyMaximization2014,tsallisConceptualInadequacyShore2015,tsallisThermodynamicsMorePowerful2014}.

In this study we can safely use both concepts. On the one hand, when we are interested solely in a consistent measure of disorder, possibly a useful mechanism for the visual system in comparing textures, there is no {\it a priori} reason to reject a form of entropy with a free parameter, particularly when the usual form is recovered for some special value of that parameter. On the other hand, while generating ensembles of textures using maximum entropy principle (e.g. as in~\cite{victorLocalImageStatistics2012}), we stick with the usual form, as convergence is already complex to characterize and we are assuming equilibrium. That is not to say maximization of other forms of entropy won't potentially lead to texture families with interesting psychophysical properties. It is an open question to understand what fundamental properties of a texture ensemble we gain (if any) and known property we loose (if any) in such a context. Therefore we propose an alternative approach to quantify disorder in texture. We consider only the BGSvN entropy for now, but we are interested in the partial contribution from each orbit, to the total entropy. We define these orbit entropies as
\begin{equation}\label{eq:partent}
S_r^b = -\sum_{i \in O^{G}_r}p_i\ln p_i,
\end{equation}
where the Boolean subscript $b$ stands for the presence of contrast reversal in $G$. The index $r$ in $S_r^b$ corresponds to the orbits $\mathcal{O}_r^b$ displayed in Table~\ref{tb:orbits}. Figure~\ref{fig:scatteraxial} shows the dispersion for the pure dihedral ($b = 0$) entropies of $G = D_4$, $\mathcal{S}_2^{0},\mathcal{S}_6^{0}$ for all textures in Figure~\ref{fig:axial}. Figure~\ref{fig:scattercomb} shows the pure dihedral entropies $\mathcal{S}_1^{0}, \mathcal{S}_5^{0}$ for a selected combination of textures shown in Figure~\ref{fig:planar}. For reference, Figure~\ref{fig:entbarplot} in Support Information displays all the information necessary to make a similar graph for all the other combinations.

\subsection*{Orbit Invariants}
 \label{sec:invs}
As shown in Table~\ref{tb:coords-orbits} there are $r = 6$ invariants to the action of the dihedral group $D_4$. We label these invariants $I^0_{\{r=1\cdots6\}}$, with the Boolean superscript $0$ indicating the absence of contrast reversal in the symmetry group. They are given in the Fourier domain by
\begin{center}
\begin{minipage}[t]{\textwidth}
{\footnotesize
\begin{align}
I_0^0 & = \bar{\alpha}
	  =    \varphi(\begin{smallmatrix} 0 & 0\\ 0 & 0 \end{smallmatrix}) \label{eq:i00}\\
I_1^0 & =    \mathbin{^{\centerdot}\gamma} +\gamma^{\centerdot}+\mathbin{_{\centerdot}\gamma} +\gamma_{\centerdot}
	  =    \varphi(\begin{smallmatrix} 1 & 0\\ 0 & 0 \end{smallmatrix})
		  +\varphi(\begin{smallmatrix} 0 & 1\\ 0 & 0 \end{smallmatrix})
		  +\varphi(\begin{smallmatrix} 0 & 0\\ 1 & 0 \end{smallmatrix})
		  +\varphi(\begin{smallmatrix} 0 & 0\\ 0 & 1 \end{smallmatrix})\label{eq:i10}\\  	
I_2^0 & =    \beta^{-} + \beta_{-} + \mathbin{_|^{}\beta}_{} +\mathbin{_{}^{}\beta}_{|}
 	  =    \varphi(\begin{smallmatrix} 1 & 1\\ 0 & 0 \end{smallmatrix})
	      +\varphi(\begin{smallmatrix} 0 & 0\\ 1 & 1 \end{smallmatrix})
	      +\varphi(\begin{smallmatrix} 1 & 0\\ 1 & 0 \end{smallmatrix}) 
	      +\varphi(\begin{smallmatrix} 0 & 1\\ 0 & 1 \end{smallmatrix})\label{eq:i20}
	      \\
I_3^0 & =    \beta_{\diagup} + \beta_{\diagdown}
	  =    \varphi(\begin{smallmatrix} 0 & 1\\ 1 & 0 \end{smallmatrix})
		  +\varphi(\begin{smallmatrix} 1 & 0\\ 0 & 1 \end{smallmatrix})\label{eq:i30}\\	
I_4^0 & =    \theta^{\ulcorner}+\theta^{\urcorner}+\theta_{\llcorner}+\theta_{\lrcorner}
	  =    \varphi(\begin{smallmatrix} 1 & 1\\ 1 & 0 \end{smallmatrix})
		  +\varphi(\begin{smallmatrix} 1 & 1\\ 0 & 1 \end{smallmatrix})
		  +\varphi(\begin{smallmatrix} 1 & 0\\ 1 & 1 \end{smallmatrix})
		  +\varphi(\begin{smallmatrix} 0 & 1\\ 1 & 1 \end{smallmatrix})\label{eq:i40}\\
I_5^0 & = \alpha = \varphi(\begin{smallmatrix} 1 & 1\\ 1 & 1 \end{smallmatrix})\label{eq:i50}.		
\end{align}
}
\end{minipage}
\end{center}
With the addition of contrast reversal $K^{2}$, the symmetry group becomes $D_4^2 = D_{4}\mtimes K^2$, we label these invariants $I^1_{\{r=1\cdots4\}}$. They are written in the original probability space as in~\cite{victorPerceptualSpaceLocal2015}
\begin{center}
\begin{minipage}[t]{\textwidth}  
{\footnotesize
\begin{align}
I_0^1 & =  I_0^0 + I_5^0
	=	 2p(\begin{smallmatrix} 0 & 0\\ 0 & 0 \end{smallmatrix})
		+2p(\begin{smallmatrix} 1 & 1\\ 0 & 0 \end{smallmatrix})
		+2p(\begin{smallmatrix} 1 & 0\\ 1 & 0 \end{smallmatrix})
		+2p(\begin{smallmatrix} 1 & 0\\ 0 & 1 \end{smallmatrix})
		+2p(\begin{smallmatrix} 0 & 1\\ 1 & 0 \end{smallmatrix})
		+2p(\begin{smallmatrix} 0 & 1\\ 0 & 1 \end{smallmatrix})
		+2p(\begin{smallmatrix} 0 & 0\\ 1 & 1 \end{smallmatrix})
		+2p(\begin{smallmatrix} 1 & 1\\ 1 & 1 \end{smallmatrix}) \label{eq:i01}
	    \\
I_1^1 & = I_1^0 +I_4^0  =  8p(\begin{smallmatrix} 0 & 0\\ 0 & 0 \end{smallmatrix})
		-8p(\begin{smallmatrix} 1 & 1\\ 1 & 1 \end{smallmatrix}) \label{eq:i11}
	     \\
I_2^1 & =4p(\begin{smallmatrix} 0 & 0\\ 0 & 0 \end{smallmatrix})
	    -4p(\begin{smallmatrix} 0 & 1\\ 1 & 0 \end{smallmatrix})
	    -4p(\begin{smallmatrix} 1 & 0\\ 0 & 1 \end{smallmatrix})
	    +4p(\begin{smallmatrix} 1 & 1\\ 1 & 1 \end{smallmatrix}) \label{eq:i21}	
		\\
I_3^1 & =2p(\begin{smallmatrix} 0 & 0\\ 0 & 0 \end{smallmatrix})
	    -2p(\begin{smallmatrix} 1 & 1\\ 0 & 0 \end{smallmatrix})
		-2p(\begin{smallmatrix} 1 & 0\\ 1 & 0 \end{smallmatrix})
		-2p(\begin{smallmatrix} 1 & 0\\ 0 & 1 \end{smallmatrix})
		-2p(\begin{smallmatrix} 0 & 1\\ 1 & 0 \end{smallmatrix})
		-2p(\begin{smallmatrix} 0 & 1\\ 0 & 1 \end{smallmatrix})
		-2p(\begin{smallmatrix} 0 & 0\\ 1 & 1 \end{smallmatrix})
		+2p(\begin{smallmatrix} 1 & 1\\ 1 & 1 \end{smallmatrix})\label{eq:i31}						
\end{align}
}
\end{minipage}
\end{center}
The index $r$ in these invariants $I_r^b$ corresponds to the index of the orbits $\mathcal{O}_r^b$ displayed in Table~\ref{tb:orbits}. The forward Fourier transform is shown in the Supporting Information, Equation~\ref{eq:jvtransform} and its inverse given in Equation~\ref{eq:jvitransform}. Figure~\ref{fig:invbarplot} shows the values for these invariants for all axial binary textures and their planar combination. This information is a reference to visualize their dispersion for all the axial textures, like the examples shown in Figure~\ref{fig:axial} and for all two combination (planar) textures, like the example shown in Figure~\ref{fig:planar}.

\subsection*{Minkowski Functionals}
The Minkowski functionals are also often called additive functionals, intrinsic (or generalized) volumes, or yet Quermassintegrals. They are either geometrical or topological global quantities that can be calculated by summing up contributions of small building blocks and exploiting their additivity. We refer the reader to our previous work showing the relevance of these functionals in explaining human discrimination performance for binary textures~\cite{barbosaLocallyCountableProperties2013}. An in depth study of the mathematical properties and review of the broad applicability in science and engineering of these functionals can be found in~\cite{michielsenIntegralgeometryMorphologicalImage2001,kapferLocalAnisotropyFluids2010,mantzUtilizingMinkowskiFunctionals2008,knufingFractalAnalysisMethods2005,schroder-turkMinkowskiTensorsAnisotropic2010,schroder-turkMinkowskiTensorShape2011,schneiderStochasticIntegralGeometry2008}. 

Considering a binary planar image as an union of disjoint convex open bodies, these functionals consist of shape measures of a flat object composed of contiguous regions (one labelled pixels) and holes (zero labelled pixels). For the case at hand, $d = 2$, the 3 Minkowski functionals are the total \emph{area} of the pattern $M_{1}$, its \emph{perimeter} $M_{2}$, and a topological concept called \emph{the Euler number}, $M_3$ or $M_{4}$ if we consider 4 or 8 nearest neighbors respectively. The Euler number is the number of connected components minus the number of holes in a pattern. For binary textures, the Minkowski functionals can be expressed explicitly in terms of the $D_4$ orbit invariants. First note (see~\cite{barbosaLocallyCountableProperties2013}) that they represent oblique directions in the perception space $V_2^{2\mtimes2}$

\begin{align}  
&M_1 =\frac{1 -\gamma}{2} \\
&M_2 -1 = - 2(\beta_{-}+\beta_{|}) \\
&M_3 -\frac{1}{16} = \frac{1}{16}
					\left[ - 4 \gamma -2\beta_{-} -2\beta_{|}
							+   \beta_{\diagup}+\beta_{\diagdown}
							+   \theta_{\ulcorner}+\theta_{\urcorner}+\theta_{\llcorner}+\theta_{\lrcorner}
							+   \alpha
					\right]\\
&M_4 +\frac{1}{16} = \frac{1}{16}
					\left[ - 4 \gamma +2\beta_{-} +2\beta_{|}
							-   \beta_{\diagup}-\beta_{\diagdown}
							+   \theta^{\ulcorner}+\theta^{\urcorner}+\theta_{\llcorner}+\theta_{\lrcorner}
							-   \alpha
					\right], 										
\end{align}
where $M_1$ is the total area, $M_2$ the perimeter and $M_{3,4}$ the Euler number for 4/8-connected neighbourhoods. This can be expressed in terms of the orbit invariants of Equations~\ref{eq:i00} to~\ref{eq:i50} as

\begin{align}
&M_1 =\frac{I_0^0 -I^0_1/4}{2}\\
&M_2 = - I^0_2 +I^0_1\\
&M_3 = \frac{1}{16}
					\left[+I^0_0 - I^0_1 -I^0_2
							+   I^0_3
							+   I^0_4
							+   I^0_5
					\right]\label{eq:x4}\\
&M_4 = \frac{1}{16}
					\left[-I^0_0 - I^0_1 +I^0_2
							-   I^0_3
							+   I^0_4
							-   I^0_5
					\right]\label{eq:x8}. 										
\end{align}

\section*{Results}
The present study was motivated by the fact that human performance in discriminating binary textures seems to rely on moments of the distribution of the Minkowski functionals -- the Euler Number, and to a lesser extent the area and perimeter. This paper carries an important mathematical result: the three Minkowski functionals for planar images can be expressed as a linear combination of the orbit invariants of the dihedral group $D_4$. With the addition of contrast reversal operation, only the \emph{total Perimeter} remains a invariant of the direct product group $D_4\mtimes K^2$. This result means that we can calculate a surrogate to the Minkowski functionals for a broader class of textures, because the orbit invariants are well defined for any neighbourhood size, lattice and number of grey levels. 

We expect this group theoretic generalization of the Minkowski functionals will have a similar prominence in explaining the human discrimination of \emph{grey-level} textures as the standard Minkowski functionals have for \emph{binary} textures. 

In this paper we presented an example of how to characterize texture salience with substantially fewer variables, namely orbit entropies; in other words, the coefficients of the Fourier transform of pixel patch probabilities over-represent the perception space. Further, this dimensionality reduction based on symmetry also helps us to create a feasible experimental design. We can use the same crowd-based software infrastructure, as used in~\cite{seamonsUnderlyingNeuralMechanisms2016}, in order to collect evidence to justify our expectations.

\section*{Discussion}

The Minkowski functionals have been associated with human discrimination performance for \emph{binary} textures both in controlled lab conditions and in a large crowd-sourced experiment~\cite{barbosaLocallyCountableProperties2013,seamonsLowerBoundNumber2015}. From our psychophysical experiments, the moments of the distribution of the area and perimeter had less predictive power of human texture discrimination compared with the Euler number. The set of binary textures defined for the correlations within $2\mtimes2$ patches seems to be described by either 4 or 6 orbit invariants. This seems to be in good agreement with the number of texture discrimination putative mechanisms used by humans~\cite{seamonsUnderlyingNeuralMechanisms2016,victorPerceptualSpaceLocal2015}. It will be interesting to see if this extends to richer textures created for more grey-levels, patch sizes and lattice types. As an example application, the introduced orbit entropies provide an economic way to represent the relationships between these textures in Figure~\ref{fig:axial} and their compositions in Figure~\ref{fig:planar}.

\section*{Conflicts of Interest}
The authors declare no conflicts of interest.

\section*{Funding}
This research was partly funded by Australian Research Council grant (LP140100763). This Australian government grant is 25\% funded by an industry partner, nuCoria Pty Ltd. There are no patents under IP under licence to nuCoria related to the ARC grant. This research was also supported by the Australian Research Council through the ARC Centre of Excellence in Vision Science (CE0561903).

\section*{Supporting Information}
\label{sec:App}

\subsection*{Groups: Conception and Perception.}

A set of transformations $S$ and a law of composition $o$, is called a \emph{group} $G = (S,o)$ if four simple properties hold. The binary operation $o$ between any two elements of $S$ must be \emph{closed} and \emph{associative}, the set $S$ must have an \emph{identity} transformation and every element of the set $S$ must have an \emph{inverse} transformation. It is our everyday experience with the Euclidean space (an infinite group) that imbue our minds with the intuition of a group of operations, as put earlier by Helmholtz and Poincar\'e,~\cite{cassirerConceptGroupTheory1944} even before the mathematical concepts involved were fully formalized and the breadth of their consequences realized. The theory of groups, together with the associated ideas of invariance and equivalence, lead to a diverse body of applications in very practical problems. There were many early studies highlighting aspects of group theory in the general study of perception~\cite{hoffmanLieAlgebraVisual1966,dodwellLieTransformationGroup1983,hoffmanHigherVisualPerception1970} and of visual discrimination of textures~\cite{caelliPerceptualAnalyzersUnderlying1978a,caelliPerceptualAnalyzersUnderlying1978,caelliPsychophysicalEvidenceGlobal1979} to name a few.

\subsection*{Group action.} A group of transformations can act on a set $X$ of objects in various ways. The group \emph{action} is a description of the mechanism used by each group element to transform an object. We can create a set of objects by colouring a single object with different colours. Consider the sequence of integer numbers of length $N$, $O = \{a_1,a_2,a_3,\cdots,a_N\}$. Consider one particular set of sequences with each number restricted to binary tokens (say black or white). In this case, objects in this set are binary sequences of length $N$.

The action of a group element on a sequence may be, for example, by permuting the positions in the sequence or by cycling through the possible colours. This may produce a new colouring or not. If this action does not change that colouring, the sequence $x$ of $X$ is called a fixed point of the corresponding group element action. All sequences $x$ fixed by $g$ form a (fixed) subspace of $X$. Conversely, distinct points related by an element $g$ of $G$ form an \emph{orbit} of $g$, and the union of all these orbits reconstruct $X$. Formally, we define the set of all fixed points by an element of $G$ by $X^g=\{x\in X|gx=x\}$ and the set of all orbits by $X/G=\{G(x) | x \in X \}$. The exact relationship between these two quantities -- the number of orbits and the cardinality of fixed subspaces will be made explicit shortly, here we just mention that these are the main constructs we need from group theory.

In this paper, the group action has an extra geometric aspect on top of the purely algebraic action described on plain sequences example above. To be concrete as possible, in this paper objects will be sequences of length $4$, rearranged in quads ($2\mtimes2$ pixel patches), but the group action still boils down to permutations or cycles, albeit also carrying the geometrical meaning of rotations or reflections on the quads when $G$ is the dihedral group.

\subsection*{Concrete group action on binary texture space.}
\label{sec:ps}

The elements of the group $G$ act on the probability of pixel patches by permuting pixel positions or cycling pixel grey-levels. In this paper the action of a generic element of $g$ of $G$ takes the following concrete form on the probability of $2\mtimes2$ pixel patches
\begin{equation}
\label{eq:base}
g\cdot p\left(\begin{matrix}\{k_1,k_2\} & \{k_3,k_4\}\\ \{k_5,k_6\} & \{k_7,k_8\} \end{matrix} \right) = p\left(\begin{matrix} k_1 & k_3\\ k_5 & k_7 \end{matrix} \right).
\end{equation}
The reason for the above (unusual) concrete definition is to mimic software implementation. Here if $k_{1}$ is one, $k_{2}$ is zero, same for the other pairs. To revert to probabilities, only the first value of the sequence is of relevance. Therefore, the combined action of the elements of the group $A = C_2^{\mtimes4}$, the cartesian product of 4 cyclic groups, on $p(\begin{smallmatrix}0 & 0\\ 0 & 0 \end{smallmatrix})$ is  
\begin{equation}
\label{eq:quads}
 \{a_1\cdots a_k\}\cdot p(\begin{smallmatrix}\{0,1\} & \{0,1\}\\ \{0,1\} & \{0,1\} \end{smallmatrix})\rightarrow\left\{p(\begin{smallmatrix}0 & 0\\ 0 & 0 \end{smallmatrix}),p(\begin{smallmatrix}1 & 0\\ 0 & 0 \end{smallmatrix}), p(\begin{smallmatrix}0 & 1\\ 0 & 0 \end{smallmatrix}),...,p(\begin{smallmatrix}1 & 1\\1 & 1\end{smallmatrix}) \right\},
 \end{equation}
and generates (iteratively) the set of all coordinates for $V_2^{2\mtimes2}$. The result of this action creates an arbitrary ordering for the coordinates of the space, following the order of the group elements in $\{a_1\cdots a_k\}$.  In the next session, we will relate this ordering with previously published notation~\cite{victorLocalImageStatistics2012} and show the implications of other particular choices.

There are two other relevant groups for this study, namely the Dihedral group $D_4$, and the group of contrast reversals. Binary contrast reversal $K^2$ is a group with just two operators, the identity and the cycling $c$ of grey-levels for all pixel positions \emph{at the same time}, in cycles notation this transformation is $c = (1,2)(3,4)(5,6)(7,8)$, and the action on pixel patch probability is

\begin{equation}
\label{eq:ci-action}
c\cdot p\left(\begin{matrix}\{k_1,k_2\} & \{k_3,k_4\}\\ \{k_5,k_6\} & \{k_7,k_8\} \end{matrix} \right) =
p\left(\begin{matrix}\{k_2,k_1\} & \{k_4,k_3\}\\ \{k_6,k_5\} & \{k_8,k_7\} \end{matrix} \right).
\end{equation}
Similarly a $\frac{\pi}{2}$ clockwise rotation, one of the 8 elements of the Dihedral group $D_4$, is given in cycles by $r = (1,3,7,5)(2,4,8,6)$, the action is
\begin{equation}
\label{eq:rot-action}
r \cdot p\left(\begin{matrix}\{k_1,k_2\} & \{k_3,k_4\}\\ \{k_5,k_6\} & \{k_7,k_8\} \end{matrix} \right) =
p\left(\begin{matrix}\{k_5,k_6\} & \{k_1,k_2\}\\ \{k_7,k_8\} & \{k_3,k_4\} \end{matrix} \right).
\end{equation}

\subsection*{Group representations.}

The abstract definition of a transformation group is connected to a concrete form by specifying a group action on an object. From the \emph{mechanics} of the group action a set of matrices can be established. These matrices \emph{represent} group elements, in the sense that the product of these matrices obeys the same multiplication table as the abstract group elements themselves. This set of matrices constitute a \emph{representation} of $G$. There are multiple representations for the same abstract group and not necessarily with the same dimension. In general, for applications (and the case here) one is interested in a specific action and a set of representations associated with it.

Many group actions will generate representations that can be decomposed into smaller irreducible or absolutely irreducible representations (depending on if the matrices are real or complex). These representations can be decomposed into a direct sum of these irreducible building block representations, called \emph{irreps}, in which the matrices for all elements of the group are in the same block diagonal form. A special representation, called the \emph{regular} representation offers such a decomposition and importantly, with all existing irreducibles displayed as repeated blocks -- their multiplicities (number of repeats) being equal to their dimension.

The decomposition into irreducibles is the central information for applications of group representation theory, to quantum mechanics, statistics, optics, engineering, ranking, encryption, and as in here, to biology and medicine~\cite{chirikjianEngineeringApplicationsNoncommutative2000,barbosaDifferentialInvariantsSymplectic2004,wolfSymmetryMultistabilityLongRange2005,golubitskySymmetryMethodsMathematical2015,vianaSymmetryStudiesRefraction2014,vianaDihedralFourierAnalysis2013,kondorRankingKernelsFourier2010a}

\subsection*{Representations of direct product of groups.}
To construct the Fourier transform of a function of group elements, we will need to be able to construct all irreducible representations of that group. In the present case, the group of interest is a direct product of smaller groups and its representations are easy to construct in terms of the representations of the smaller groups.

Matrix representations of the direct product of two groups $\mathcal{H}_1$ and $\mathcal{H}_2$ are simply the direct product of the matrix representations of each group
\begin{equation}
\rho^{ij}\left(G=\mathcal{H}_1\mtimes\mathcal{H}_2\right) = \rho^{i}\left( \mathcal{H}_1 \right) \otimes \rho^{j}\left(\mathcal{H}_2\right)
\end{equation}
The matrix elements for this representation are given by
\begin{equation}
\label{eq:dirpro}
\rho^{ij}_{pq;rs}(g) = \rho^{i}_{pr}(h_1)\rho^{j}_{qs}(h_2), \quad \text{with} \quad g=h_1h_2.
\end{equation}
Now we use the fact that the direct product of cyclic groups is also a cyclic group. These groups are commutative (Abelian) and thus have unidimensional irreps. Therefore we use the index $i$ for each irrep $\rho$ for the single number $\rho^{i}(g)$ \emph{representing} the element $g$ in the $i$-th representation of the group $G$. For the cyclic group, $C_n = \{1,g,g^2,\cdots,g^{n-1}\}$, there is a simple closed form for all the $n$ irreducible representations, namely
\begin{equation}
\rho^{p}(g^m) = e^{\frac{m(p-1)}{n}2\pi i}.
\end{equation}
Note that in the formula above we wrote the factor $p-1$ to insure conformity with the usual convention that the first irreducible representation is the identity representation, the number $1$ for all elements of $A$, satisfying the multiplication rules of the group trivially. We rewrite this as
\begin{equation}
\label{eq:rep}
\rho^{s}(g^a) = \omega^{-sa}\quad \mathrm{where} \quad s = p-1, \quad a = m, \quad \omega = e^{-\frac{2\pi i}{n}},
\end{equation}
where we use symmetric indexes $s,a \in \{0,...,n-1\}$.

\subsection*{Fourier Transform on Finite Groups.}
\label{sec:ftfg}
The Fourier transform of a function $f(g), g \in G$ is given by a collection of coefficients indexed by a particular irreducible representation $\rho^{i}$ of $G$,
\begin{equation}
\label{eq:fft}
\hat{f}(\rho^{i}) = \sum_{g\in G}f(g)\rho^{i}(g^{-1}).
\end{equation}
The inverse transformation is given by
\begin{equation}
\label{eq:ifft}
f(g) = \sum_{j=1}^{\alpha} d_j \mathrm{trace} [\hat{f}(\rho^{j})\rho^{j}(g)],
\end{equation}
where $\alpha$ is the the number of distinct irreps of $G$ and $d_j$ is the multiplicity of the irrep $\rho^j$ in the decomposition of the regular representation of $G$.

\subsection*{Fourier Transform on Abelian groups.}
\label{sec:ftag}
Using the representations of $C_n$ in Equation~\ref{eq:rep}, applying recursively Equation~\ref{eq:dirpro} for the representations of the direct product of cyclic groups, and using the definition of the Fourier transform given by Equation~\ref{eq:fft},  we arrive at
\begin{equation}
\label{eq:afft}
\hat{f}(\mathbf{s}) = \sum_{\{a_k=0\}_{k\in\{1,2,\cdots,N\}}}^{n-1}f(\mathbf{a})\omega^{\mathbf{s}\cdot\mathbf{a}}.
\end{equation}

For binary images we have $n=2$ and $N=4$, and $\omega = -1$. By relabelling the generic function $f$ for probability $p$ of $2\mtimes 2$ pixel patches $p$, and writing $\varphi = \hat{p}$, we have

\begin{equation*}
\hat{f}(\mathbf{s}) = \varphi(\begin{smallmatrix} s_1 & s_2\\ s_3 & s_4 \end{smallmatrix}), \quad
f(\mathbf{a}) = p(\begin{smallmatrix} a_1 & a_2\\ a_3 & a_4\end{smallmatrix}),
\end{equation*}

and

\begin{align}\label{eq:ftonC4}
\varphi(\begin{smallmatrix} s_1 & s_2\\ s_3 & s_4 \end{smallmatrix})&=
\sum_{a_{1\cdots4}\in\{0,1\}}p(\begin{smallmatrix} a_1 & a_2\\ a_3 & a_4\end{smallmatrix})(-1)^{\mathbf{s}\cdot\mathbf{a}}.
\end{align}
We show in the next section this forward Fourier Transform explicitly in Equations~\ref{eq:jvtransform} and its inverse in Equation~\ref{eq:jvitransform}. We refer to the space spanned by the basis $\varphi(\begin{smallmatrix} s_1 & s_2\\ s_3 & s_4 \end{smallmatrix})$ as the~\emph{perception space}. 

\subsection*{The Orbit Counting Theorem}\label{sec:oct}
In this section we show the relationship between the number of orbits and the average number of fixed points. This relationship has been referred to as the Cauchy-Frobenious Theorem, the orbit counting Theorem or, colloquially, as the Lemma that is not Burnside's. We state first two useful lemmas without proof.

\begin{lemma}[Lagrange's Theorem]\label{lem:Lagrange}
For any finite group $G$, the order of every subgroup $H$ of $G$ divides the order of $G$
\begin{equation}
n|H| = |G|.
\end{equation}
for some integer $n$.
\end{lemma}

Define $G_x = Stab(x) = \{g\in G|g*x = x\}$ as the stabilizer of $x$ by $G$. Also, let $G(x) = Orb(x) = \{y\in X| y = g*x\}$ be the orbit of $x$, i.e. all possible destinations of $x$ under the group action. We have
\begin{lemma}[Orbit-Stabilizer Theorem] The size of the orbit of $x$ is the order of the group $G$ divided by the size of the stabilizer of $x$. \label{lem:os}
\begin{equation}
|G(x)| = \frac{|G|} {|G_x|}.
\end{equation}
\end{lemma}

\begin{theorem}[Orbit Counting Theorem] The number of orbits of $G$ acting on a set $X$ is the average of the number of fixed points $x$ by the action of the elements of $G$.
\begin{equation}\label{eq:numberoforbits}
 |X/G|=\frac{1}{|G|}\sum_{g\in G}|X^g|
\end{equation}
\end{theorem}
\begin{proof}
Define the set of all fixed points by an
element of $G$ by $X^g=\{x\in X|gx=x\}$ and the set of all orbits by $X/G=
\{G(x) | x \in X \}$ and using lemmas~\ref{lem:Lagrange} and~\ref{lem:os} we have
\begin{equation*}
\begin{aligned}
\sum_{g\in G}|X^g|&=\sum_{g\in G}(\sum_{x:gx=x} 1)=\sum_{x\in X}(\sum_{g:gx=x} 1)\\
                &=\sum_{x\in X}|G_x|=\sum_{x\in X} \frac{|G|}{|G(x)|}   \quad \hookleftarrow\text{ using Lemma \ref{lem:os}}\\
                &=|G|\sum_{\omega \in X/G}\sum_{x \in \omega} (\frac{1}{|\omega|}) \quad \hookleftarrow\text{using Lemma \ref{lem:Lagrange}} \\                
                &=|G||X/G|, 
\end{aligned}
\end{equation*}
\end{proof}
where we indicate any orbit $G(x)$ by the dummy variable $\omega$, and used the fact if $x_1$ and $x_2$ are in the same orbit then $|G(x_1)|$ = $|G(x_2)|$ is the size of their common orbit $|\omega|$. So
\begin{equation*}
\sum_{x \in \omega} (\frac{1}{|\omega|}) = \frac{\overbrace{1+1+\cdots,1}^{|\omega|\text{-times}}}{|\omega|}=1.
\end{equation*}

\subsection*{Fourier Transform of \texorpdfstring{$G = C_2^{\mtimes 4}$}}
\label{sec:ifft}
The explicit form of the Fourier transform in Equation~\ref{eq:ftonC4} is given by the Equations~\ref{eq:jvtransform} below.
\begin{equation}\label{eq:jvtransform}
\hspace{-.5cm}
\resizebox{.97\hsize}{!}{$
\begin{aligned}
\varphi(\begin{smallmatrix} 0 & 0\\ 0 & 0 \end{smallmatrix})&=p(\begin{smallmatrix} 0 & 0\\ 0 & 0 \end{smallmatrix})+p(\begin{smallmatrix} 1 & 0\\ 0 & 0 \end{smallmatrix})+p(\begin{smallmatrix} 0 & 1\\ 0 & 0 \end{smallmatrix})+p(\begin{smallmatrix} 1 & 1\\ 0 & 0 \end{smallmatrix})+p(\begin{smallmatrix} 0 & 0\\ 1 & 0 \end{smallmatrix})+p(\begin{smallmatrix} 1 & 0\\ 1 & 0 \end{smallmatrix})+p(\begin{smallmatrix} 0 & 1\\ 1 & 0 \end{smallmatrix})+p(\begin{smallmatrix} 1 & 1\\ 1 & 0 \end{smallmatrix})+p(\begin{smallmatrix} 0 & 0\\ 0 & 1 \end{smallmatrix})+p(\begin{smallmatrix} 1 & 0\\ 0 & 1 \end{smallmatrix})+p(\begin{smallmatrix} 0 & 1\\ 0 & 1 \end{smallmatrix})+p(\begin{smallmatrix} 1 & 1\\ 0 & 1 \end{smallmatrix})+p(\begin{smallmatrix} 0 & 0\\ 1 & 1 \end{smallmatrix})+p(\begin{smallmatrix} 1 & 0\\ 1 & 1 \end{smallmatrix})+p(\begin{smallmatrix} 0 & 1\\ 1 & 1 \end{smallmatrix})+p(\begin{smallmatrix} 1 & 1\\ 1 & 1 \end{smallmatrix})\\
\varphi(\begin{smallmatrix} 1 & 0\\ 0 & 0 \end{smallmatrix})&=p(\begin{smallmatrix} 0 & 0\\ 0 & 0 \end{smallmatrix})-p(\begin{smallmatrix} 1 & 0\\ 0 & 0 \end{smallmatrix})+p(\begin{smallmatrix} 0 & 1\\ 0 & 0 \end{smallmatrix})-p(\begin{smallmatrix} 1 & 1\\ 0 & 0 \end{smallmatrix})+p(\begin{smallmatrix} 0 & 0\\ 1 & 0 \end{smallmatrix})-p(\begin{smallmatrix} 1 & 0\\ 1 & 0 \end{smallmatrix})+p(\begin{smallmatrix} 0 & 1\\ 1 & 0 \end{smallmatrix})-p(\begin{smallmatrix} 1 & 1\\ 1 & 0 \end{smallmatrix})+p(\begin{smallmatrix} 0 & 0\\ 0 & 1 \end{smallmatrix})-p(\begin{smallmatrix} 1 & 0\\ 0 & 1 \end{smallmatrix})+p(\begin{smallmatrix} 0 & 1\\ 0 & 1 \end{smallmatrix})-p(\begin{smallmatrix} 1 & 1\\ 0 & 1 \end{smallmatrix})+p(\begin{smallmatrix} 0 & 0\\ 1 & 1 \end{smallmatrix})-p(\begin{smallmatrix} 1 & 0\\ 1 & 1 \end{smallmatrix})+p(\begin{smallmatrix} 0 & 1\\ 1 & 1 \end{smallmatrix})-p(\begin{smallmatrix} 1 & 1\\ 1 & 1 \end{smallmatrix})\\
\varphi(\begin{smallmatrix} 0 & 1\\ 0 & 0 \end{smallmatrix})&=p(\begin{smallmatrix} 0 & 0\\ 0 & 0 \end{smallmatrix})+p(\begin{smallmatrix} 1 & 0\\ 0 & 0 \end{smallmatrix})-p(\begin{smallmatrix} 0 & 1\\ 0 & 0 \end{smallmatrix})-p(\begin{smallmatrix} 1 & 1\\ 0 & 0 \end{smallmatrix})+p(\begin{smallmatrix} 0 & 0\\ 1 & 0 \end{smallmatrix})+p(\begin{smallmatrix} 1 & 0\\ 1 & 0 \end{smallmatrix})-p(\begin{smallmatrix} 0 & 1\\ 1 & 0 \end{smallmatrix})-p(\begin{smallmatrix} 1 & 1\\ 1 & 0 \end{smallmatrix})+p(\begin{smallmatrix} 0 & 0\\ 0 & 1 \end{smallmatrix})+p(\begin{smallmatrix} 1 & 0\\ 0 & 1 \end{smallmatrix})-p(\begin{smallmatrix} 0 & 1\\ 0 & 1 \end{smallmatrix})-p(\begin{smallmatrix} 1 & 1\\ 0 & 1 \end{smallmatrix})+p(\begin{smallmatrix} 0 & 0\\ 1 & 1 \end{smallmatrix})+p(\begin{smallmatrix} 1 & 0\\ 1 & 1 \end{smallmatrix})-p(\begin{smallmatrix} 0 & 1\\ 1 & 1 \end{smallmatrix})-p(\begin{smallmatrix} 1 & 1\\ 1 & 1 \end{smallmatrix})\\
\varphi(\begin{smallmatrix} 1 & 1\\ 0 & 0 \end{smallmatrix})&=p(\begin{smallmatrix} 0 & 0\\ 0 & 0 \end{smallmatrix})-p(\begin{smallmatrix} 1 & 0\\ 0 & 0 \end{smallmatrix})-p(\begin{smallmatrix} 0 & 1\\ 0 & 0 \end{smallmatrix})+p(\begin{smallmatrix} 1 & 1\\ 0 & 0 \end{smallmatrix})+p(\begin{smallmatrix} 0 & 0\\ 1 & 0 \end{smallmatrix})-p(\begin{smallmatrix} 1 & 0\\ 1 & 0 \end{smallmatrix})-p(\begin{smallmatrix} 0 & 1\\ 1 & 0 \end{smallmatrix})+p(\begin{smallmatrix} 1 & 1\\ 1 & 0 \end{smallmatrix})+p(\begin{smallmatrix} 0 & 0\\ 0 & 1 \end{smallmatrix})-p(\begin{smallmatrix} 1 & 0\\ 0 & 1 \end{smallmatrix})-p(\begin{smallmatrix} 0 & 1\\ 0 & 1 \end{smallmatrix})+p(\begin{smallmatrix} 1 & 1\\ 0 & 1 \end{smallmatrix})+p(\begin{smallmatrix} 0 & 0\\ 1 & 1 \end{smallmatrix})-p(\begin{smallmatrix} 1 & 0\\ 1 & 1 \end{smallmatrix})-p(\begin{smallmatrix} 0 & 1\\ 1 & 1 \end{smallmatrix})+p(\begin{smallmatrix} 1 & 1\\ 1 & 1 \end{smallmatrix})\\
\varphi(\begin{smallmatrix} 0 & 0\\ 1 & 0 \end{smallmatrix})&=p(\begin{smallmatrix} 0 & 0\\ 0 & 0 \end{smallmatrix})+p(\begin{smallmatrix} 1 & 0\\ 0 & 0 \end{smallmatrix})+p(\begin{smallmatrix} 0 & 1\\ 0 & 0 \end{smallmatrix})+p(\begin{smallmatrix} 1 & 1\\ 0 & 0 \end{smallmatrix})-p(\begin{smallmatrix} 0 & 0\\ 1 & 0 \end{smallmatrix})-p(\begin{smallmatrix} 1 & 0\\ 1 & 0 \end{smallmatrix})-p(\begin{smallmatrix} 0 & 1\\ 1 & 0 \end{smallmatrix})-p(\begin{smallmatrix} 1 & 1\\ 1 & 0 \end{smallmatrix})+p(\begin{smallmatrix} 0 & 0\\ 0 & 1 \end{smallmatrix})+p(\begin{smallmatrix} 1 & 0\\ 0 & 1 \end{smallmatrix})+p(\begin{smallmatrix} 0 & 1\\ 0 & 1 \end{smallmatrix})+p(\begin{smallmatrix} 1 & 1\\ 0 & 1 \end{smallmatrix})-p(\begin{smallmatrix} 0 & 0\\ 1 & 1 \end{smallmatrix})-p(\begin{smallmatrix} 1 & 0\\ 1 & 1 \end{smallmatrix})-p(\begin{smallmatrix} 0 & 1\\ 1 & 1 \end{smallmatrix})-p(\begin{smallmatrix} 1 & 1\\ 1 & 1 \end{smallmatrix})\\
\varphi(\begin{smallmatrix} 1 & 0\\ 1 & 0 \end{smallmatrix})&=p(\begin{smallmatrix} 0 & 0\\ 0 & 0 \end{smallmatrix})-p(\begin{smallmatrix} 1 & 0\\ 0 & 0 \end{smallmatrix})+p(\begin{smallmatrix} 0 & 1\\ 0 & 0 \end{smallmatrix})-p(\begin{smallmatrix} 1 & 1\\ 0 & 0 \end{smallmatrix})-p(\begin{smallmatrix} 0 & 0\\ 1 & 0 \end{smallmatrix})+p(\begin{smallmatrix} 1 & 0\\ 1 & 0 \end{smallmatrix})-p(\begin{smallmatrix} 0 & 1\\ 1 & 0 \end{smallmatrix})+p(\begin{smallmatrix} 1 & 1\\ 1 & 0 \end{smallmatrix})+p(\begin{smallmatrix} 0 & 0\\ 0 & 1 \end{smallmatrix})-p(\begin{smallmatrix} 1 & 0\\ 0 & 1 \end{smallmatrix})+p(\begin{smallmatrix} 0 & 1\\ 0 & 1 \end{smallmatrix})-p(\begin{smallmatrix} 1 & 1\\ 0 & 1 \end{smallmatrix})-p(\begin{smallmatrix} 0 & 0\\ 1 & 1 \end{smallmatrix})+p(\begin{smallmatrix} 1 & 0\\ 1 & 1 \end{smallmatrix})-p(\begin{smallmatrix} 0 & 1\\ 1 & 1 \end{smallmatrix})+p(\begin{smallmatrix} 1 & 1\\ 1 & 1 \end{smallmatrix})\\
\varphi(\begin{smallmatrix} 0 & 1\\ 1 & 0 \end{smallmatrix})&=p(\begin{smallmatrix} 0 & 0\\ 0 & 0 \end{smallmatrix})+p(\begin{smallmatrix} 1 & 0\\ 0 & 0 \end{smallmatrix})-p(\begin{smallmatrix} 0 & 1\\ 0 & 0 \end{smallmatrix})-p(\begin{smallmatrix} 1 & 1\\ 0 & 0 \end{smallmatrix})-p(\begin{smallmatrix} 0 & 0\\ 1 & 0 \end{smallmatrix})-p(\begin{smallmatrix} 1 & 0\\ 1 & 0 \end{smallmatrix})+p(\begin{smallmatrix} 0 & 1\\ 1 & 0 \end{smallmatrix})+p(\begin{smallmatrix} 1 & 1\\ 1 & 0 \end{smallmatrix})+p(\begin{smallmatrix} 0 & 0\\ 0 & 1 \end{smallmatrix})+p(\begin{smallmatrix} 1 & 0\\ 0 & 1 \end{smallmatrix})-p(\begin{smallmatrix} 0 & 1\\ 0 & 1 \end{smallmatrix})-p(\begin{smallmatrix} 1 & 1\\ 0 & 1 \end{smallmatrix})-p(\begin{smallmatrix} 0 & 0\\ 1 & 1 \end{smallmatrix})-p(\begin{smallmatrix} 1 & 0\\ 1 & 1 \end{smallmatrix})+p(\begin{smallmatrix} 0 & 1\\ 1 & 1 \end{smallmatrix})+p(\begin{smallmatrix} 1 & 1\\ 1 & 1 \end{smallmatrix})\\
\varphi(\begin{smallmatrix} 1 & 1\\ 1 & 0 \end{smallmatrix})&=p(\begin{smallmatrix} 0 & 0\\ 0 & 0 \end{smallmatrix})-p(\begin{smallmatrix} 1 & 0\\ 0 & 0 \end{smallmatrix})-p(\begin{smallmatrix} 0 & 1\\ 0 & 0 \end{smallmatrix})+p(\begin{smallmatrix} 1 & 1\\ 0 & 0 \end{smallmatrix})-p(\begin{smallmatrix} 0 & 0\\ 1 & 0 \end{smallmatrix})+p(\begin{smallmatrix} 1 & 0\\ 1 & 0 \end{smallmatrix})+p(\begin{smallmatrix} 0 & 1\\ 1 & 0 \end{smallmatrix})-p(\begin{smallmatrix} 1 & 1\\ 1 & 0 \end{smallmatrix})+p(\begin{smallmatrix} 0 & 0\\ 0 & 1 \end{smallmatrix})-p(\begin{smallmatrix} 1 & 0\\ 0 & 1 \end{smallmatrix})-p(\begin{smallmatrix} 0 & 1\\ 0 & 1 \end{smallmatrix})+p(\begin{smallmatrix} 1 & 1\\ 0 & 1 \end{smallmatrix})-p(\begin{smallmatrix} 0 & 0\\ 1 & 1 \end{smallmatrix})+p(\begin{smallmatrix} 1 & 0\\ 1 & 1 \end{smallmatrix})+p(\begin{smallmatrix} 0 & 1\\ 1 & 1 \end{smallmatrix})-p(\begin{smallmatrix} 1 & 1\\ 1 & 1 \end{smallmatrix})\\
\varphi(\begin{smallmatrix} 0 & 0\\ 0 & 1 \end{smallmatrix})&=p(\begin{smallmatrix} 0 & 0\\ 0 & 0 \end{smallmatrix})+p(\begin{smallmatrix} 1 & 0\\ 0 & 0 \end{smallmatrix})+p(\begin{smallmatrix} 0 & 1\\ 0 & 0 \end{smallmatrix})+p(\begin{smallmatrix} 1 & 1\\ 0 & 0 \end{smallmatrix})+p(\begin{smallmatrix} 0 & 0\\ 1 & 0 \end{smallmatrix})+p(\begin{smallmatrix} 1 & 0\\ 1 & 0 \end{smallmatrix})+p(\begin{smallmatrix} 0 & 1\\ 1 & 0 \end{smallmatrix})+p(\begin{smallmatrix} 1 & 1\\ 1 & 0 \end{smallmatrix})-p(\begin{smallmatrix} 0 & 0\\ 0 & 1 \end{smallmatrix})-p(\begin{smallmatrix} 1 & 0\\ 0 & 1 \end{smallmatrix})-p(\begin{smallmatrix} 0 & 1\\ 0 & 1 \end{smallmatrix})-p(\begin{smallmatrix} 1 & 1\\ 0 & 1 \end{smallmatrix})-p(\begin{smallmatrix} 0 & 0\\ 1 & 1 \end{smallmatrix})-p(\begin{smallmatrix} 1 & 0\\ 1 & 1 \end{smallmatrix})-p(\begin{smallmatrix} 0 & 1\\ 1 & 1 \end{smallmatrix})-p(\begin{smallmatrix} 1 & 1\\ 1 & 1 \end{smallmatrix})\\
\varphi(\begin{smallmatrix} 1 & 0\\ 0 & 1 \end{smallmatrix})&=p(\begin{smallmatrix} 0 & 0\\ 0 & 0 \end{smallmatrix})-p(\begin{smallmatrix} 1 & 0\\ 0 & 0 \end{smallmatrix})+p(\begin{smallmatrix} 0 & 1\\ 0 & 0 \end{smallmatrix})-p(\begin{smallmatrix} 1 & 1\\ 0 & 0 \end{smallmatrix})+p(\begin{smallmatrix} 0 & 0\\ 1 & 0 \end{smallmatrix})-p(\begin{smallmatrix} 1 & 0\\ 1 & 0 \end{smallmatrix})+p(\begin{smallmatrix} 0 & 1\\ 1 & 0 \end{smallmatrix})-p(\begin{smallmatrix} 1 & 1\\ 1 & 0 \end{smallmatrix})-p(\begin{smallmatrix} 0 & 0\\ 0 & 1 \end{smallmatrix})+p(\begin{smallmatrix} 1 & 0\\ 0 & 1 \end{smallmatrix})-p(\begin{smallmatrix} 0 & 1\\ 0 & 1 \end{smallmatrix})+p(\begin{smallmatrix} 1 & 1\\ 0 & 1 \end{smallmatrix})-p(\begin{smallmatrix} 0 & 0\\ 1 & 1 \end{smallmatrix})+p(\begin{smallmatrix} 1 & 0\\ 1 & 1 \end{smallmatrix})-p(\begin{smallmatrix} 0 & 1\\ 1 & 1 \end{smallmatrix})+p(\begin{smallmatrix} 1 & 1\\ 1 & 1 \end{smallmatrix})\\
\varphi(\begin{smallmatrix} 0 & 1\\ 0 & 1 \end{smallmatrix})&=p(\begin{smallmatrix} 0 & 0\\ 0 & 0 \end{smallmatrix})+p(\begin{smallmatrix} 1 & 0\\ 0 & 0 \end{smallmatrix})-p(\begin{smallmatrix} 0 & 1\\ 0 & 0 \end{smallmatrix})-p(\begin{smallmatrix} 1 & 1\\ 0 & 0 \end{smallmatrix})+p(\begin{smallmatrix} 0 & 0\\ 1 & 0 \end{smallmatrix})+p(\begin{smallmatrix} 1 & 0\\ 1 & 0 \end{smallmatrix})-p(\begin{smallmatrix} 0 & 1\\ 1 & 0 \end{smallmatrix})-p(\begin{smallmatrix} 1 & 1\\ 1 & 0 \end{smallmatrix})-p(\begin{smallmatrix} 0 & 0\\ 0 & 1 \end{smallmatrix})-p(\begin{smallmatrix} 1 & 0\\ 0 & 1 \end{smallmatrix})+p(\begin{smallmatrix} 0 & 1\\ 0 & 1 \end{smallmatrix})+p(\begin{smallmatrix} 1 & 1\\ 0 & 1 \end{smallmatrix})-p(\begin{smallmatrix} 0 & 0\\ 1 & 1 \end{smallmatrix})-p(\begin{smallmatrix} 1 & 0\\ 1 & 1 \end{smallmatrix})+p(\begin{smallmatrix} 0 & 1\\ 1 & 1 \end{smallmatrix})+p(\begin{smallmatrix} 1 & 1\\ 1 & 1 \end{smallmatrix})\\
\varphi(\begin{smallmatrix} 1 & 1\\ 0 & 1 \end{smallmatrix})&=p(\begin{smallmatrix} 0 & 0\\ 0 & 0 \end{smallmatrix})-p(\begin{smallmatrix} 1 & 0\\ 0 & 0 \end{smallmatrix})-p(\begin{smallmatrix} 0 & 1\\ 0 & 0 \end{smallmatrix})+p(\begin{smallmatrix} 1 & 1\\ 0 & 0 \end{smallmatrix})+p(\begin{smallmatrix} 0 & 0\\ 1 & 0 \end{smallmatrix})-p(\begin{smallmatrix} 1 & 0\\ 1 & 0 \end{smallmatrix})-p(\begin{smallmatrix} 0 & 1\\ 1 & 0 \end{smallmatrix})+p(\begin{smallmatrix} 1 & 1\\ 1 & 0 \end{smallmatrix})-p(\begin{smallmatrix} 0 & 0\\ 0 & 1 \end{smallmatrix})+p(\begin{smallmatrix} 1 & 0\\ 0 & 1 \end{smallmatrix})+p(\begin{smallmatrix} 0 & 1\\ 0 & 1 \end{smallmatrix})-p(\begin{smallmatrix} 1 & 1\\ 0 & 1 \end{smallmatrix})-p(\begin{smallmatrix} 0 & 0\\ 1 & 1 \end{smallmatrix})+p(\begin{smallmatrix} 1 & 0\\ 1 & 1 \end{smallmatrix})+p(\begin{smallmatrix} 0 & 1\\ 1 & 1 \end{smallmatrix})-p(\begin{smallmatrix} 1 & 1\\ 1 & 1 \end{smallmatrix})\\
\varphi(\begin{smallmatrix} 0 & 0\\ 1 & 1 \end{smallmatrix})&=p(\begin{smallmatrix} 0 & 0\\ 0 & 0 \end{smallmatrix})+p(\begin{smallmatrix} 1 & 0\\ 0 & 0 \end{smallmatrix})+p(\begin{smallmatrix} 0 & 1\\ 0 & 0 \end{smallmatrix})+p(\begin{smallmatrix} 1 & 1\\ 0 & 0 \end{smallmatrix})-p(\begin{smallmatrix} 0 & 0\\ 1 & 0 \end{smallmatrix})-p(\begin{smallmatrix} 1 & 0\\ 1 & 0 \end{smallmatrix})-p(\begin{smallmatrix} 0 & 1\\ 1 & 0 \end{smallmatrix})-p(\begin{smallmatrix} 1 & 1\\ 1 & 0 \end{smallmatrix})-p(\begin{smallmatrix} 0 & 0\\ 0 & 1 \end{smallmatrix})-p(\begin{smallmatrix} 1 & 0\\ 0 & 1 \end{smallmatrix})-p(\begin{smallmatrix} 0 & 1\\ 0 & 1 \end{smallmatrix})-p(\begin{smallmatrix} 1 & 1\\ 0 & 1 \end{smallmatrix})+p(\begin{smallmatrix} 0 & 0\\ 1 & 1 \end{smallmatrix})+p(\begin{smallmatrix} 1 & 0\\ 1 & 1 \end{smallmatrix})+p(\begin{smallmatrix} 0 & 1\\ 1 & 1 \end{smallmatrix})+p(\begin{smallmatrix} 1 & 1\\ 1 & 1 \end{smallmatrix})\\
\varphi(\begin{smallmatrix} 1 & 0\\ 1 & 1 \end{smallmatrix})&=p(\begin{smallmatrix} 0 & 0\\ 0 & 0 \end{smallmatrix})-p(\begin{smallmatrix} 1 & 0\\ 0 & 0 \end{smallmatrix})+p(\begin{smallmatrix} 0 & 1\\ 0 & 0 \end{smallmatrix})-p(\begin{smallmatrix} 1 & 1\\ 0 & 0 \end{smallmatrix})-p(\begin{smallmatrix} 0 & 0\\ 1 & 0 \end{smallmatrix})+p(\begin{smallmatrix} 1 & 0\\ 1 & 0 \end{smallmatrix})-p(\begin{smallmatrix} 0 & 1\\ 1 & 0 \end{smallmatrix})+p(\begin{smallmatrix} 1 & 1\\ 1 & 0 \end{smallmatrix})-p(\begin{smallmatrix} 0 & 0\\ 0 & 1 \end{smallmatrix})+p(\begin{smallmatrix} 1 & 0\\ 0 & 1 \end{smallmatrix})-p(\begin{smallmatrix} 0 & 1\\ 0 & 1 \end{smallmatrix})+p(\begin{smallmatrix} 1 & 1\\ 0 & 1 \end{smallmatrix})+p(\begin{smallmatrix} 0 & 0\\ 1 & 1 \end{smallmatrix})-p(\begin{smallmatrix} 1 & 0\\ 1 & 1 \end{smallmatrix})+p(\begin{smallmatrix} 0 & 1\\ 1 & 1 \end{smallmatrix})-p(\begin{smallmatrix} 1 & 1\\ 1 & 1 \end{smallmatrix})\\
\varphi(\begin{smallmatrix} 0 & 1\\ 1 & 1 \end{smallmatrix})&=p(\begin{smallmatrix} 0 & 0\\ 0 & 0 \end{smallmatrix})+p(\begin{smallmatrix} 1 & 0\\ 0 & 0 \end{smallmatrix})-p(\begin{smallmatrix} 0 & 1\\ 0 & 0 \end{smallmatrix})-p(\begin{smallmatrix} 1 & 1\\ 0 & 0 \end{smallmatrix})-p(\begin{smallmatrix} 0 & 0\\ 1 & 0 \end{smallmatrix})-p(\begin{smallmatrix} 1 & 0\\ 1 & 0 \end{smallmatrix})+p(\begin{smallmatrix} 0 & 1\\ 1 & 0 \end{smallmatrix})+p(\begin{smallmatrix} 1 & 1\\ 1 & 0 \end{smallmatrix})-p(\begin{smallmatrix} 0 & 0\\ 0 & 1 \end{smallmatrix})-p(\begin{smallmatrix} 1 & 0\\ 0 & 1 \end{smallmatrix})+p(\begin{smallmatrix} 0 & 1\\ 0 & 1 \end{smallmatrix})+p(\begin{smallmatrix} 1 & 1\\ 0 & 1 \end{smallmatrix})+p(\begin{smallmatrix} 0 & 0\\ 1 & 1 \end{smallmatrix})+p(\begin{smallmatrix} 1 & 0\\ 1 & 1 \end{smallmatrix})-p(\begin{smallmatrix} 0 & 1\\ 1 & 1 \end{smallmatrix})-p(\begin{smallmatrix} 1 & 1\\ 1 & 1 \end{smallmatrix})\\
\varphi(\begin{smallmatrix} 1 & 1\\ 1 & 1 \end{smallmatrix})&=p(\begin{smallmatrix} 0 & 0\\ 0 & 0 \end{smallmatrix})-p(\begin{smallmatrix} 1 & 0\\ 0 & 0 \end{smallmatrix})-p(\begin{smallmatrix} 0 & 1\\ 0 & 0 \end{smallmatrix})+p(\begin{smallmatrix} 1 & 1\\ 0 & 0 \end{smallmatrix})-p(\begin{smallmatrix} 0 & 0\\ 1 & 0 \end{smallmatrix})+p(\begin{smallmatrix} 1 & 0\\ 1 & 0 \end{smallmatrix})+p(\begin{smallmatrix} 0 & 1\\ 1 & 0 \end{smallmatrix})-p(\begin{smallmatrix} 1 & 1\\ 1 & 0 \end{smallmatrix})-p(\begin{smallmatrix} 0 & 0\\ 0 & 1 \end{smallmatrix})+p(\begin{smallmatrix} 1 & 0\\ 0 & 1 \end{smallmatrix})+p(\begin{smallmatrix} 0 & 1\\ 0 & 1 \end{smallmatrix})-p(\begin{smallmatrix} 1 & 1\\ 0 & 1 \end{smallmatrix})+p(\begin{smallmatrix} 0 & 0\\ 1 & 1 \end{smallmatrix})-p(\begin{smallmatrix} 1 & 0\\ 1 & 1 \end{smallmatrix})-p(\begin{smallmatrix} 0 & 1\\ 1 & 1 \end{smallmatrix})+p(\begin{smallmatrix} 1 & 1\\ 1 & 1 \end{smallmatrix})\\
\end{aligned}$}
\end{equation}

Some of the coordinates in the Equations~\ref{eq:jvtransform} are identical for transforms of functions representing the probability of $2\mtimes2$ pixel patches, e.g. $\{\varphi(\begin{smallmatrix} 0 & 0\\ 1 & 1 \end{smallmatrix}) = \varphi(\begin{smallmatrix} 1 & 1\\ 0 & 0 \end{smallmatrix}),\varphi(\begin{smallmatrix} 1 & 0\\ 0 & 0 \end{smallmatrix}) = \varphi(\begin{smallmatrix} 0 & 1\\ 0 & 0 \end{smallmatrix}) \}$. Similarly, it is obvious that $\varphi(\begin{smallmatrix} 0 & 0\\ 0 & 0 \end{smallmatrix}) = 1$, but in general it is a free coordinate and we will keep it as such because it impacts upon the results of the symmetry analysis.
The inverse of the Fourier transform in Equations~\ref{eq:jvtransform} is given in below by Equation~\ref{eq:jvitransform}.
\begin{equation}\label{eq:jvitransform}
\hspace{-.56cm}
\resizebox{.96\hsize}{!}{$
\begin{aligned}
16p(\begin{smallmatrix} 0 & 0\\ 0 & 0 \end{smallmatrix})&=\varphi(\begin{smallmatrix} 0 & 0\\ 0 & 0 \end{smallmatrix})+\varphi(\begin{smallmatrix} 1 & 0\\ 0 & 0 \end{smallmatrix})+\varphi(\begin{smallmatrix} 0 & 1\\ 0 & 0 \end{smallmatrix})+\varphi(\begin{smallmatrix} 1 & 1\\ 0 & 0 \end{smallmatrix})+\varphi(\begin{smallmatrix} 0 & 0\\ 1 & 0 \end{smallmatrix})+\varphi(\begin{smallmatrix} 1 & 0\\ 1 & 0 \end{smallmatrix})+\varphi(\begin{smallmatrix} 0 & 1\\ 1 & 0 \end{smallmatrix})+\varphi(\begin{smallmatrix} 1 & 1\\ 1 & 0 \end{smallmatrix})+\varphi(\begin{smallmatrix} 0 & 0\\ 0 & 1 \end{smallmatrix})+\varphi(\begin{smallmatrix} 1 & 0\\ 0 & 1 \end{smallmatrix})+\varphi(\begin{smallmatrix} 0 & 1\\ 0 & 1 \end{smallmatrix})+\varphi(\begin{smallmatrix} 1 & 1\\ 0 & 1 \end{smallmatrix})+\varphi(\begin{smallmatrix} 0 & 0\\ 1 & 1 \end{smallmatrix})+\varphi(\begin{smallmatrix} 1 & 0\\ 1 & 1 \end{smallmatrix})+\varphi(\begin{smallmatrix} 0 & 1\\ 1 & 1 \end{smallmatrix})+\varphi(\begin{smallmatrix} 1 & 1\\ 1 & 1 \end{smallmatrix})\\
16p(\begin{smallmatrix} 1 & 0\\ 0 & 0 \end{smallmatrix})&=\varphi(\begin{smallmatrix} 0 & 0\\ 0 & 0 \end{smallmatrix})-\varphi(\begin{smallmatrix} 1 & 0\\ 0 & 0 \end{smallmatrix})+\varphi(\begin{smallmatrix} 0 & 1\\ 0 & 0 \end{smallmatrix})-\varphi(\begin{smallmatrix} 1 & 1\\ 0 & 0 \end{smallmatrix})+\varphi(\begin{smallmatrix} 0 & 0\\ 1 & 0 \end{smallmatrix})-\varphi(\begin{smallmatrix} 1 & 0\\ 1 & 0 \end{smallmatrix})+\varphi(\begin{smallmatrix} 0 & 1\\ 1 & 0 \end{smallmatrix})-\varphi(\begin{smallmatrix} 1 & 1\\ 1 & 0 \end{smallmatrix})+\varphi(\begin{smallmatrix} 0 & 0\\ 0 & 1 \end{smallmatrix})-\varphi(\begin{smallmatrix} 1 & 0\\ 0 & 1 \end{smallmatrix})+\varphi(\begin{smallmatrix} 0 & 1\\ 0 & 1 \end{smallmatrix})-\varphi(\begin{smallmatrix} 1 & 1\\ 0 & 1 \end{smallmatrix})+\varphi(\begin{smallmatrix} 0 & 0\\ 1 & 1 \end{smallmatrix})-\varphi(\begin{smallmatrix} 1 & 0\\ 1 & 1 \end{smallmatrix})+\varphi(\begin{smallmatrix} 0 & 1\\ 1 & 1 \end{smallmatrix})-\varphi(\begin{smallmatrix} 1 & 1\\ 1 & 1 \end{smallmatrix})\\
16p(\begin{smallmatrix} 0 & 1\\ 0 & 0 \end{smallmatrix})&=\varphi(\begin{smallmatrix} 0 & 0\\ 0 & 0 \end{smallmatrix})+\varphi(\begin{smallmatrix} 1 & 0\\ 0 & 0 \end{smallmatrix})-\varphi(\begin{smallmatrix} 0 & 1\\ 0 & 0 \end{smallmatrix})-\varphi(\begin{smallmatrix} 1 & 1\\ 0 & 0 \end{smallmatrix})+\varphi(\begin{smallmatrix} 0 & 0\\ 1 & 0 \end{smallmatrix})+\varphi(\begin{smallmatrix} 1 & 0\\ 1 & 0 \end{smallmatrix})-\varphi(\begin{smallmatrix} 0 & 1\\ 1 & 0 \end{smallmatrix})-\varphi(\begin{smallmatrix} 1 & 1\\ 1 & 0 \end{smallmatrix})+\varphi(\begin{smallmatrix} 0 & 0\\ 0 & 1 \end{smallmatrix})+\varphi(\begin{smallmatrix} 1 & 0\\ 0 & 1 \end{smallmatrix})-\varphi(\begin{smallmatrix} 0 & 1\\ 0 & 1 \end{smallmatrix})-\varphi(\begin{smallmatrix} 1 & 1\\ 0 & 1 \end{smallmatrix})+\varphi(\begin{smallmatrix} 0 & 0\\ 1 & 1 \end{smallmatrix})+\varphi(\begin{smallmatrix} 1 & 0\\ 1 & 1 \end{smallmatrix})-\varphi(\begin{smallmatrix} 0 & 1\\ 1 & 1 \end{smallmatrix})-\varphi(\begin{smallmatrix} 1 & 1\\ 1 & 1 \end{smallmatrix})\\
16p(\begin{smallmatrix} 1 & 1\\ 0 & 0 \end{smallmatrix})&=\varphi(\begin{smallmatrix} 0 & 0\\ 0 & 0 \end{smallmatrix})-\varphi(\begin{smallmatrix} 1 & 0\\ 0 & 0 \end{smallmatrix})-\varphi(\begin{smallmatrix} 0 & 1\\ 0 & 0 \end{smallmatrix})+\varphi(\begin{smallmatrix} 1 & 1\\ 0 & 0 \end{smallmatrix})+\varphi(\begin{smallmatrix} 0 & 0\\ 1 & 0 \end{smallmatrix})-\varphi(\begin{smallmatrix} 1 & 0\\ 1 & 0 \end{smallmatrix})-\varphi(\begin{smallmatrix} 0 & 1\\ 1 & 0 \end{smallmatrix})+\varphi(\begin{smallmatrix} 1 & 1\\ 1 & 0 \end{smallmatrix})+\varphi(\begin{smallmatrix} 0 & 0\\ 0 & 1 \end{smallmatrix})-\varphi(\begin{smallmatrix} 1 & 0\\ 0 & 1 \end{smallmatrix})-\varphi(\begin{smallmatrix} 0 & 1\\ 0 & 1 \end{smallmatrix})+\varphi(\begin{smallmatrix} 1 & 1\\ 0 & 1 \end{smallmatrix})+\varphi(\begin{smallmatrix} 0 & 0\\ 1 & 1 \end{smallmatrix})-\varphi(\begin{smallmatrix} 1 & 0\\ 1 & 1 \end{smallmatrix})-\varphi(\begin{smallmatrix} 0 & 1\\ 1 & 1 \end{smallmatrix})+\varphi(\begin{smallmatrix} 1 & 1\\ 1 & 1 \end{smallmatrix})\\
16p(\begin{smallmatrix} 0 & 0\\ 1 & 0 \end{smallmatrix})&=\varphi(\begin{smallmatrix} 0 & 0\\ 0 & 0 \end{smallmatrix})+\varphi(\begin{smallmatrix} 1 & 0\\ 0 & 0 \end{smallmatrix})+\varphi(\begin{smallmatrix} 0 & 1\\ 0 & 0 \end{smallmatrix})+\varphi(\begin{smallmatrix} 1 & 1\\ 0 & 0 \end{smallmatrix})-\varphi(\begin{smallmatrix} 0 & 0\\ 1 & 0 \end{smallmatrix})-\varphi(\begin{smallmatrix} 1 & 0\\ 1 & 0 \end{smallmatrix})-\varphi(\begin{smallmatrix} 0 & 1\\ 1 & 0 \end{smallmatrix})-\varphi(\begin{smallmatrix} 1 & 1\\ 1 & 0 \end{smallmatrix})+\varphi(\begin{smallmatrix} 0 & 0\\ 0 & 1 \end{smallmatrix})+\varphi(\begin{smallmatrix} 1 & 0\\ 0 & 1 \end{smallmatrix})+\varphi(\begin{smallmatrix} 0 & 1\\ 0 & 1 \end{smallmatrix})+\varphi(\begin{smallmatrix} 1 & 1\\ 0 & 1 \end{smallmatrix})-\varphi(\begin{smallmatrix} 0 & 0\\ 1 & 1 \end{smallmatrix})-\varphi(\begin{smallmatrix} 1 & 0\\ 1 & 1 \end{smallmatrix})-\varphi(\begin{smallmatrix} 0 & 1\\ 1 & 1 \end{smallmatrix})-\varphi(\begin{smallmatrix} 1 & 1\\ 1 & 1 \end{smallmatrix})\\
16p(\begin{smallmatrix} 1 & 0\\ 1 & 0 \end{smallmatrix})&=\varphi(\begin{smallmatrix} 0 & 0\\ 0 & 0 \end{smallmatrix})-\varphi(\begin{smallmatrix} 1 & 0\\ 0 & 0 \end{smallmatrix})+\varphi(\begin{smallmatrix} 0 & 1\\ 0 & 0 \end{smallmatrix})-\varphi(\begin{smallmatrix} 1 & 1\\ 0 & 0 \end{smallmatrix})-\varphi(\begin{smallmatrix} 0 & 0\\ 1 & 0 \end{smallmatrix})+\varphi(\begin{smallmatrix} 1 & 0\\ 1 & 0 \end{smallmatrix})-\varphi(\begin{smallmatrix} 0 & 1\\ 1 & 0 \end{smallmatrix})+\varphi(\begin{smallmatrix} 1 & 1\\ 1 & 0 \end{smallmatrix})+\varphi(\begin{smallmatrix} 0 & 0\\ 0 & 1 \end{smallmatrix})-\varphi(\begin{smallmatrix} 1 & 0\\ 0 & 1 \end{smallmatrix})+\varphi(\begin{smallmatrix} 0 & 1\\ 0 & 1 \end{smallmatrix})-\varphi(\begin{smallmatrix} 1 & 1\\ 0 & 1 \end{smallmatrix})-\varphi(\begin{smallmatrix} 0 & 0\\ 1 & 1 \end{smallmatrix})+\varphi(\begin{smallmatrix} 1 & 0\\ 1 & 1 \end{smallmatrix})-\varphi(\begin{smallmatrix} 0 & 1\\ 1 & 1 \end{smallmatrix})+\varphi(\begin{smallmatrix} 1 & 1\\ 1 & 1 \end{smallmatrix})\\
16p(\begin{smallmatrix} 0 & 1\\ 1 & 0 \end{smallmatrix})&=\varphi(\begin{smallmatrix} 0 & 0\\ 0 & 0 \end{smallmatrix})+\varphi(\begin{smallmatrix} 1 & 0\\ 0 & 0 \end{smallmatrix})-\varphi(\begin{smallmatrix} 0 & 1\\ 0 & 0 \end{smallmatrix})-\varphi(\begin{smallmatrix} 1 & 1\\ 0 & 0 \end{smallmatrix})-\varphi(\begin{smallmatrix} 0 & 0\\ 1 & 0 \end{smallmatrix})-\varphi(\begin{smallmatrix} 1 & 0\\ 1 & 0 \end{smallmatrix})+\varphi(\begin{smallmatrix} 0 & 1\\ 1 & 0 \end{smallmatrix})+\varphi(\begin{smallmatrix} 1 & 1\\ 1 & 0 \end{smallmatrix})+\varphi(\begin{smallmatrix} 0 & 0\\ 0 & 1 \end{smallmatrix})+\varphi(\begin{smallmatrix} 1 & 0\\ 0 & 1 \end{smallmatrix})-\varphi(\begin{smallmatrix} 0 & 1\\ 0 & 1 \end{smallmatrix})-\varphi(\begin{smallmatrix} 1 & 1\\ 0 & 1 \end{smallmatrix})-\varphi(\begin{smallmatrix} 0 & 0\\ 1 & 1 \end{smallmatrix})-\varphi(\begin{smallmatrix} 1 & 0\\ 1 & 1 \end{smallmatrix})+\varphi(\begin{smallmatrix} 0 & 1\\ 1 & 1 \end{smallmatrix})+\varphi(\begin{smallmatrix} 1 & 1\\ 1 & 1 \end{smallmatrix})\\
16p(\begin{smallmatrix} 1 & 1\\ 1 & 0 \end{smallmatrix})&=\varphi(\begin{smallmatrix} 0 & 0\\ 0 & 0 \end{smallmatrix})-\varphi(\begin{smallmatrix} 1 & 0\\ 0 & 0 \end{smallmatrix})-\varphi(\begin{smallmatrix} 0 & 1\\ 0 & 0 \end{smallmatrix})+\varphi(\begin{smallmatrix} 1 & 1\\ 0 & 0 \end{smallmatrix})-\varphi(\begin{smallmatrix} 0 & 0\\ 1 & 0 \end{smallmatrix})+\varphi(\begin{smallmatrix} 1 & 0\\ 1 & 0 \end{smallmatrix})+\varphi(\begin{smallmatrix} 0 & 1\\ 1 & 0 \end{smallmatrix})-\varphi(\begin{smallmatrix} 1 & 1\\ 1 & 0 \end{smallmatrix})+\varphi(\begin{smallmatrix} 0 & 0\\ 0 & 1 \end{smallmatrix})-\varphi(\begin{smallmatrix} 1 & 0\\ 0 & 1 \end{smallmatrix})-\varphi(\begin{smallmatrix} 0 & 1\\ 0 & 1 \end{smallmatrix})+\varphi(\begin{smallmatrix} 1 & 1\\ 0 & 1 \end{smallmatrix})-\varphi(\begin{smallmatrix} 0 & 0\\ 1 & 1 \end{smallmatrix})+\varphi(\begin{smallmatrix} 1 & 0\\ 1 & 1 \end{smallmatrix})+\varphi(\begin{smallmatrix} 0 & 1\\ 1 & 1 \end{smallmatrix})-\varphi(\begin{smallmatrix} 1 & 1\\ 1 & 1 \end{smallmatrix})\\
16p(\begin{smallmatrix} 0 & 0\\ 0 & 1 \end{smallmatrix})&=\varphi(\begin{smallmatrix} 0 & 0\\ 0 & 0 \end{smallmatrix})+\varphi(\begin{smallmatrix} 1 & 0\\ 0 & 0 \end{smallmatrix})+\varphi(\begin{smallmatrix} 0 & 1\\ 0 & 0 \end{smallmatrix})+\varphi(\begin{smallmatrix} 1 & 1\\ 0 & 0 \end{smallmatrix})+\varphi(\begin{smallmatrix} 0 & 0\\ 1 & 0 \end{smallmatrix})+\varphi(\begin{smallmatrix} 1 & 0\\ 1 & 0 \end{smallmatrix})+\varphi(\begin{smallmatrix} 0 & 1\\ 1 & 0 \end{smallmatrix})+\varphi(\begin{smallmatrix} 1 & 1\\ 1 & 0 \end{smallmatrix})-\varphi(\begin{smallmatrix} 0 & 0\\ 0 & 1 \end{smallmatrix})-\varphi(\begin{smallmatrix} 1 & 0\\ 0 & 1 \end{smallmatrix})-\varphi(\begin{smallmatrix} 0 & 1\\ 0 & 1 \end{smallmatrix})-\varphi(\begin{smallmatrix} 1 & 1\\ 0 & 1 \end{smallmatrix})-\varphi(\begin{smallmatrix} 0 & 0\\ 1 & 1 \end{smallmatrix})-\varphi(\begin{smallmatrix} 1 & 0\\ 1 & 1 \end{smallmatrix})-\varphi(\begin{smallmatrix} 0 & 1\\ 1 & 1 \end{smallmatrix})-\varphi(\begin{smallmatrix} 1 & 1\\ 1 & 1 \end{smallmatrix})\\
16p(\begin{smallmatrix} 1 & 0\\ 0 & 1 \end{smallmatrix})&=\varphi(\begin{smallmatrix} 0 & 0\\ 0 & 0 \end{smallmatrix})-\varphi(\begin{smallmatrix} 1 & 0\\ 0 & 0 \end{smallmatrix})+\varphi(\begin{smallmatrix} 0 & 1\\ 0 & 0 \end{smallmatrix})-\varphi(\begin{smallmatrix} 1 & 1\\ 0 & 0 \end{smallmatrix})+\varphi(\begin{smallmatrix} 0 & 0\\ 1 & 0 \end{smallmatrix})-\varphi(\begin{smallmatrix} 1 & 0\\ 1 & 0 \end{smallmatrix})+\varphi(\begin{smallmatrix} 0 & 1\\ 1 & 0 \end{smallmatrix})-\varphi(\begin{smallmatrix} 1 & 1\\ 1 & 0 \end{smallmatrix})-\varphi(\begin{smallmatrix} 0 & 0\\ 0 & 1 \end{smallmatrix})+\varphi(\begin{smallmatrix} 1 & 0\\ 0 & 1 \end{smallmatrix})-\varphi(\begin{smallmatrix} 0 & 1\\ 0 & 1 \end{smallmatrix})+\varphi(\begin{smallmatrix} 1 & 1\\ 0 & 1 \end{smallmatrix})-\varphi(\begin{smallmatrix} 0 & 0\\ 1 & 1 \end{smallmatrix})+\varphi(\begin{smallmatrix} 1 & 0\\ 1 & 1 \end{smallmatrix})-\varphi(\begin{smallmatrix} 0 & 1\\ 1 & 1 \end{smallmatrix})+\varphi(\begin{smallmatrix} 1 & 1\\ 1 & 1 \end{smallmatrix})\\
16p(\begin{smallmatrix} 0 & 1\\ 0 & 1 \end{smallmatrix})&=\varphi(\begin{smallmatrix} 0 & 0\\ 0 & 0 \end{smallmatrix})+\varphi(\begin{smallmatrix} 1 & 0\\ 0 & 0 \end{smallmatrix})-\varphi(\begin{smallmatrix} 0 & 1\\ 0 & 0 \end{smallmatrix})-\varphi(\begin{smallmatrix} 1 & 1\\ 0 & 0 \end{smallmatrix})+\varphi(\begin{smallmatrix} 0 & 0\\ 1 & 0 \end{smallmatrix})+\varphi(\begin{smallmatrix} 1 & 0\\ 1 & 0 \end{smallmatrix})-\varphi(\begin{smallmatrix} 0 & 1\\ 1 & 0 \end{smallmatrix})-\varphi(\begin{smallmatrix} 1 & 1\\ 1 & 0 \end{smallmatrix})-\varphi(\begin{smallmatrix} 0 & 0\\ 0 & 1 \end{smallmatrix})-\varphi(\begin{smallmatrix} 1 & 0\\ 0 & 1 \end{smallmatrix})+\varphi(\begin{smallmatrix} 0 & 1\\ 0 & 1 \end{smallmatrix})+\varphi(\begin{smallmatrix} 1 & 1\\ 0 & 1 \end{smallmatrix})-\varphi(\begin{smallmatrix} 0 & 0\\ 1 & 1 \end{smallmatrix})-\varphi(\begin{smallmatrix} 1 & 0\\ 1 & 1 \end{smallmatrix})+\varphi(\begin{smallmatrix} 0 & 1\\ 1 & 1 \end{smallmatrix})+\varphi(\begin{smallmatrix} 1 & 1\\ 1 & 1 \end{smallmatrix})\\
16p(\begin{smallmatrix} 1 & 1\\ 0 & 1 \end{smallmatrix})&=\varphi(\begin{smallmatrix} 0 & 0\\ 0 & 0 \end{smallmatrix})-\varphi(\begin{smallmatrix} 1 & 0\\ 0 & 0 \end{smallmatrix})-\varphi(\begin{smallmatrix} 0 & 1\\ 0 & 0 \end{smallmatrix})+\varphi(\begin{smallmatrix} 1 & 1\\ 0 & 0 \end{smallmatrix})+\varphi(\begin{smallmatrix} 0 & 0\\ 1 & 0 \end{smallmatrix})-\varphi(\begin{smallmatrix} 1 & 0\\ 1 & 0 \end{smallmatrix})-\varphi(\begin{smallmatrix} 0 & 1\\ 1 & 0 \end{smallmatrix})+\varphi(\begin{smallmatrix} 1 & 1\\ 1 & 0 \end{smallmatrix})-\varphi(\begin{smallmatrix} 0 & 0\\ 0 & 1 \end{smallmatrix})+\varphi(\begin{smallmatrix} 1 & 0\\ 0 & 1 \end{smallmatrix})+\varphi(\begin{smallmatrix} 0 & 1\\ 0 & 1 \end{smallmatrix})-\varphi(\begin{smallmatrix} 1 & 1\\ 0 & 1 \end{smallmatrix})-\varphi(\begin{smallmatrix} 0 & 0\\ 1 & 1 \end{smallmatrix})+\varphi(\begin{smallmatrix} 1 & 0\\ 1 & 1 \end{smallmatrix})+\varphi(\begin{smallmatrix} 0 & 1\\ 1 & 1 \end{smallmatrix})-\varphi(\begin{smallmatrix} 1 & 1\\ 1 & 1 \end{smallmatrix})\\
16p(\begin{smallmatrix} 0 & 0\\ 1 & 1 \end{smallmatrix})&=\varphi(\begin{smallmatrix} 0 & 0\\ 0 & 0 \end{smallmatrix})+\varphi(\begin{smallmatrix} 1 & 0\\ 0 & 0 \end{smallmatrix})+\varphi(\begin{smallmatrix} 0 & 1\\ 0 & 0 \end{smallmatrix})+\varphi(\begin{smallmatrix} 1 & 1\\ 0 & 0 \end{smallmatrix})-\varphi(\begin{smallmatrix} 0 & 0\\ 1 & 0 \end{smallmatrix})-\varphi(\begin{smallmatrix} 1 & 0\\ 1 & 0 \end{smallmatrix})-\varphi(\begin{smallmatrix} 0 & 1\\ 1 & 0 \end{smallmatrix})-\varphi(\begin{smallmatrix} 1 & 1\\ 1 & 0 \end{smallmatrix})-\varphi(\begin{smallmatrix} 0 & 0\\ 0 & 1 \end{smallmatrix})-\varphi(\begin{smallmatrix} 1 & 0\\ 0 & 1 \end{smallmatrix})-\varphi(\begin{smallmatrix} 0 & 1\\ 0 & 1 \end{smallmatrix})-\varphi(\begin{smallmatrix} 1 & 1\\ 0 & 1 \end{smallmatrix})+\varphi(\begin{smallmatrix} 0 & 0\\ 1 & 1 \end{smallmatrix})+\varphi(\begin{smallmatrix} 1 & 0\\ 1 & 1 \end{smallmatrix})+\varphi(\begin{smallmatrix} 0 & 1\\ 1 & 1 \end{smallmatrix})+\varphi(\begin{smallmatrix} 1 & 1\\ 1 & 1 \end{smallmatrix})\\
16p(\begin{smallmatrix} 1 & 0\\ 1 & 1 \end{smallmatrix})&=\varphi(\begin{smallmatrix} 0 & 0\\ 0 & 0 \end{smallmatrix})-\varphi(\begin{smallmatrix} 1 & 0\\ 0 & 0 \end{smallmatrix})+\varphi(\begin{smallmatrix} 0 & 1\\ 0 & 0 \end{smallmatrix})-\varphi(\begin{smallmatrix} 1 & 1\\ 0 & 0 \end{smallmatrix})-\varphi(\begin{smallmatrix} 0 & 0\\ 1 & 0 \end{smallmatrix})+\varphi(\begin{smallmatrix} 1 & 0\\ 1 & 0 \end{smallmatrix})-\varphi(\begin{smallmatrix} 0 & 1\\ 1 & 0 \end{smallmatrix})+\varphi(\begin{smallmatrix} 1 & 1\\ 1 & 0 \end{smallmatrix})-\varphi(\begin{smallmatrix} 0 & 0\\ 0 & 1 \end{smallmatrix})+\varphi(\begin{smallmatrix} 1 & 0\\ 0 & 1 \end{smallmatrix})-\varphi(\begin{smallmatrix} 0 & 1\\ 0 & 1 \end{smallmatrix})+\varphi(\begin{smallmatrix} 1 & 1\\ 0 & 1 \end{smallmatrix})+\varphi(\begin{smallmatrix} 0 & 0\\ 1 & 1 \end{smallmatrix})-\varphi(\begin{smallmatrix} 1 & 0\\ 1 & 1 \end{smallmatrix})+\varphi(\begin{smallmatrix} 0 & 1\\ 1 & 1 \end{smallmatrix})-\varphi(\begin{smallmatrix} 1 & 1\\ 1 & 1 \end{smallmatrix})\\
16p(\begin{smallmatrix} 0 & 1\\ 1 & 1 \end{smallmatrix})&=\varphi(\begin{smallmatrix} 0 & 0\\ 0 & 0 \end{smallmatrix})+\varphi(\begin{smallmatrix} 1 & 0\\ 0 & 0 \end{smallmatrix})-\varphi(\begin{smallmatrix} 0 & 1\\ 0 & 0 \end{smallmatrix})-\varphi(\begin{smallmatrix} 1 & 1\\ 0 & 0 \end{smallmatrix})-\varphi(\begin{smallmatrix} 0 & 0\\ 1 & 0 \end{smallmatrix})-\varphi(\begin{smallmatrix} 1 & 0\\ 1 & 0 \end{smallmatrix})+\varphi(\begin{smallmatrix} 0 & 1\\ 1 & 0 \end{smallmatrix})+\varphi(\begin{smallmatrix} 1 & 1\\ 1 & 0 \end{smallmatrix})-\varphi(\begin{smallmatrix} 0 & 0\\ 0 & 1 \end{smallmatrix})-\varphi(\begin{smallmatrix} 1 & 0\\ 0 & 1 \end{smallmatrix})+\varphi(\begin{smallmatrix} 0 & 1\\ 0 & 1 \end{smallmatrix})+\varphi(\begin{smallmatrix} 1 & 1\\ 0 & 1 \end{smallmatrix})+\varphi(\begin{smallmatrix} 0 & 0\\ 1 & 1 \end{smallmatrix})+\varphi(\begin{smallmatrix} 1 & 0\\ 1 & 1 \end{smallmatrix})-\varphi(\begin{smallmatrix} 0 & 1\\ 1 & 1 \end{smallmatrix})-\varphi(\begin{smallmatrix} 1 & 1\\ 1 & 1 \end{smallmatrix})\\
16p(\begin{smallmatrix} 1 & 1\\ 1 & 1 \end{smallmatrix})&=\varphi(\begin{smallmatrix} 0 & 0\\ 0 & 0 \end{smallmatrix})-\varphi(\begin{smallmatrix} 1 & 0\\ 0 & 0 \end{smallmatrix})-\varphi(\begin{smallmatrix} 0 & 1\\ 0 & 0 \end{smallmatrix})+\varphi(\begin{smallmatrix} 1 & 1\\ 0 & 0 \end{smallmatrix})-\varphi(\begin{smallmatrix} 0 & 0\\ 1 & 0 \end{smallmatrix})+\varphi(\begin{smallmatrix} 1 & 0\\ 1 & 0 \end{smallmatrix})+\varphi(\begin{smallmatrix} 0 & 1\\ 1 & 0 \end{smallmatrix})-\varphi(\begin{smallmatrix} 1 & 1\\ 1 & 0 \end{smallmatrix})-\varphi(\begin{smallmatrix} 0 & 0\\ 0 & 1 \end{smallmatrix})+\varphi(\begin{smallmatrix} 1 & 0\\ 0 & 1 \end{smallmatrix})+\varphi(\begin{smallmatrix} 0 & 1\\ 0 & 1 \end{smallmatrix})-\varphi(\begin{smallmatrix} 1 & 1\\ 0 & 1 \end{smallmatrix})+\varphi(\begin{smallmatrix} 0 & 0\\ 1 & 1 \end{smallmatrix})-\varphi(\begin{smallmatrix} 1 & 0\\ 1 & 1 \end{smallmatrix})-\varphi(\begin{smallmatrix} 0 & 1\\ 1 & 1 \end{smallmatrix})+\varphi(\begin{smallmatrix} 1 & 1\\ 1 & 1 \end{smallmatrix})\\
\end{aligned}
$}
\end{equation}

\subsection*{Orbit entropy and invariants of texture mixtures}
Figure~\ref{fig:entbarplot} and Figure~\ref{fig:invbarplot} shows the values for orbit entropy and invariants for all axial binary textures and their planar combination. This information is a reference to visualize their dispersion for all the axial textures, like the examples shown in Figure~\ref{fig:scatteraxial}  and for all two combination. (planar) textures, like the example shown in Figure~\ref{fig:scattercomb}.

\begin{figure}[htb]
\hspace{-1.8cm}
\centering 
\includegraphics[trim = 1cm 4cm 3cm 4cm, clip, scale = 0.9]{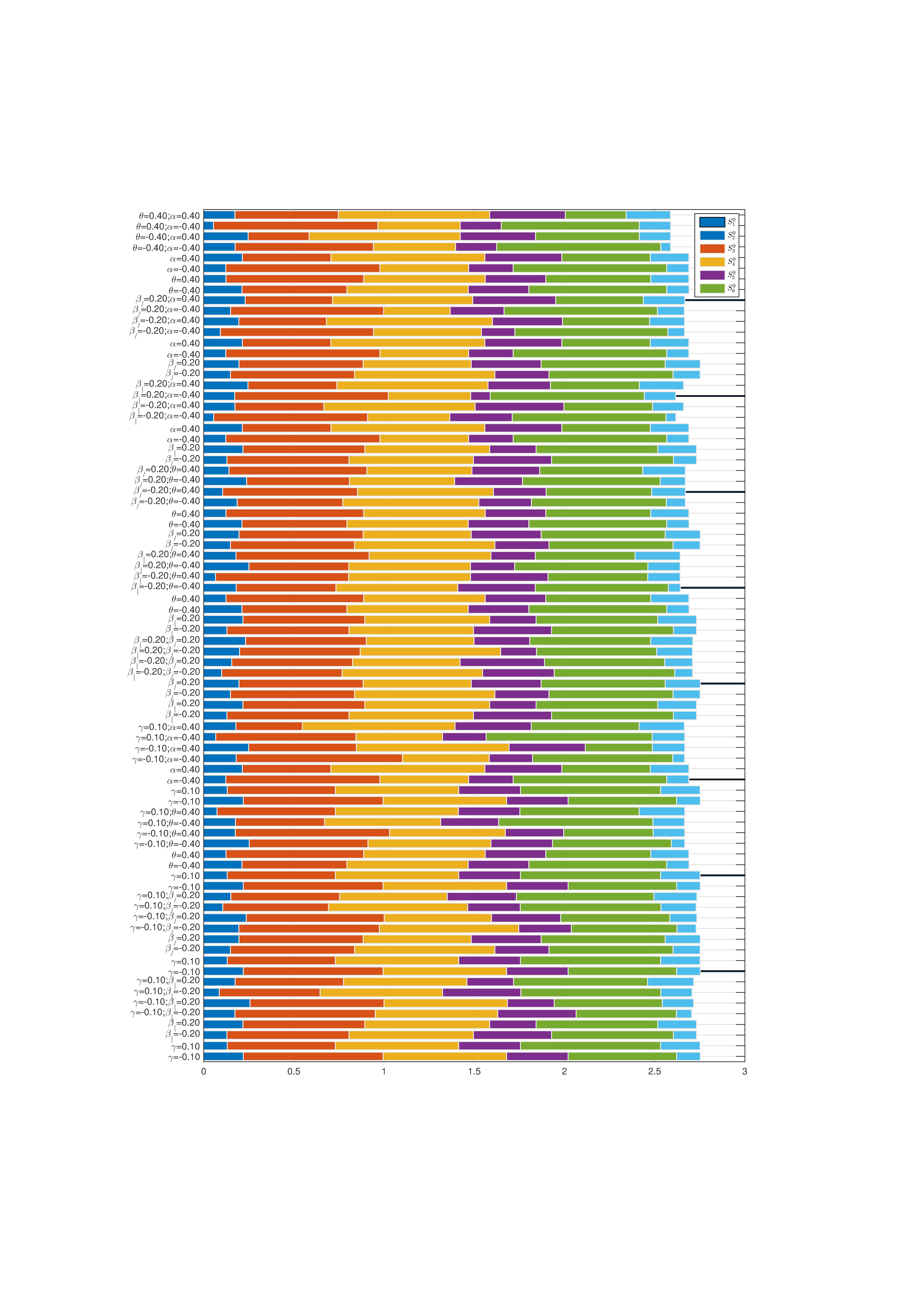}
\caption{Orbit entropies for all axial and planar textures, see Figure~\ref{fig:axial} and Figure~\ref{fig:planar}\label{fig:entbarplot}}
\end{figure}

\begin{figure}[htb]
\hspace{-1.8cm}
\centering
\includegraphics[trim = 1cm 4cm 2cm 4cm, clip, scale = 0.9]{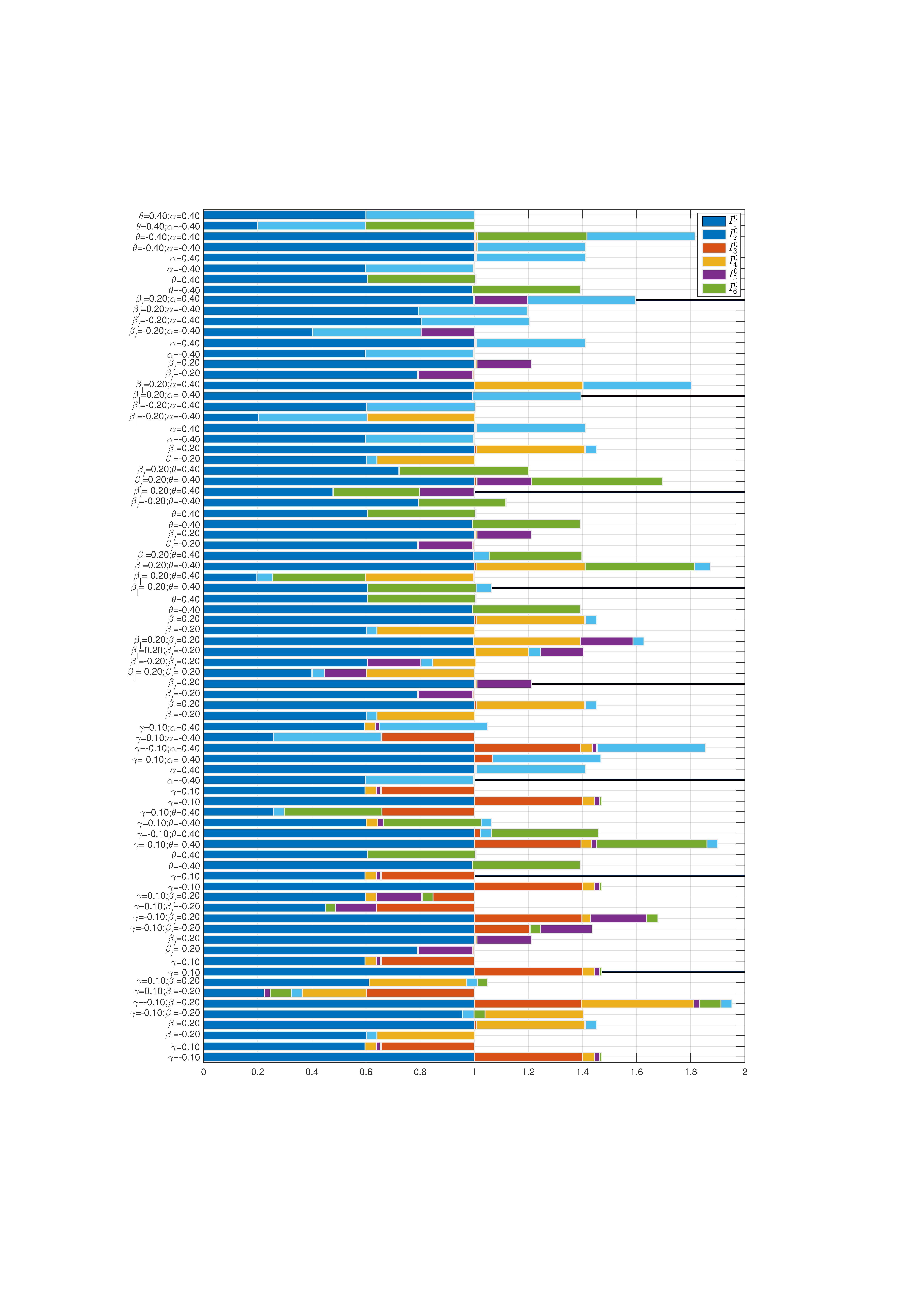}
\caption{Orbit invariants for all axial and planar textures, see Figure~\ref{fig:axial} and Figure~\ref{fig:planar}. \label{fig:invbarplot}}
\end{figure}

\subsection*{Orbits for 3-level, $2\mtimes2$ textures.}
Here summarize the extension of our analysis to 3 levels, which is still compact enough to be displayed in concrete form. In future work we have interest in these $2x2$ patches up to 4 levels and for $3x3$ patches up to 9 levels, allowing all positions to have unique values.   

\begin{table}[htbp]
\centering\vspace{0.5cm}
\begin{tabular}{ll}
a & $(id)$\\
b & $(1,4,10,7)(2,5,11,8)(3,6,12,9)$ \\
c & $(1,10)(2,11)(3,12)$ \\
d & $(1,4)(2,5)(3,6)(7,10)(8,11)(9,12)$ \\
e & $(1,7)(2,8)(3,9)(4,10)(5,11)(6,12)$ \\
f & $(1,10)(2,11)(3,12)(4,7)(5,8)(6,9)$ \\
g & $(1,7,10,4)(2,8,11,5)(3,9,12,6)$ \\
h & $(4,7)(5,8)(6,9)$ \\
a* & $(1,3)(4,6)(7,9)(10,12)$ \\
b* & $(1,6,10,9)(2,5,11,8)(3,4,12,7)$ \\
c* & $(1,12)(2,11)(3,10)(4,6)(7,9)$ \\
d* & $(1,6)(2,5)(3,4)(7,12)(8,11)(9,10)$ \\
e* & $(1,9)(2,8)(3,7)(4,12)(5,11)(6,10)$ \\
f* & $(1,12)(2,11)(3,10)(4,9)(5,8)(6,7)$ \\
g* & $(1,9,10,6)(2,8,11,5)(3,7,12,4)$ \\
h* & $(1,3)(4,9)(5,8)(6,7)(10,12)$ \\
\end{tabular}\vspace{0.5cm}
\caption{\label{tb:3let2el} The elements of $D_4^3 = D_4 \mtimes K_3$, written in cycle notation. See the orbits defined in Table~\ref{tb:3orbs}.}
\end{table}

\begin{table}
\centering\vspace{-3cm}
{\fontfamily{cmtt} \selectfont \tiny
\resizebox{.85\hsize}{!}{
\begin{tabular}{llllllllllllllllll}
ind & quad & a & b & c & d & e & f & g & h & a* & b* & c* & d* & e* & f* & g* & h* \\ \hline\hline
$81$ & \scalebox{0.8}{$\varphi\begin{bsmallmatrix}
2 & 2 \\
2 & 2
\end{bsmallmatrix}$} & $81$ & $81$ & $81$ & $81$ & $81$ & $81$ & $81$ & $81$ & $1$ & $1$ & $1$ & $1$ & $1$ & $1$ & $1$ & $1$ \\
$1$ & \scalebox{0.8}{$\varphi\begin{bsmallmatrix}
0 & 0 \\
0 & 0
\end{bsmallmatrix}$} & $1$ & $1$ & $1$ & $1$ & $1$ & $1$ & $1$ & $1$ & $81$ & $81$ & $81$ & $81$ & $81$ & $81$ & $81$ & $81$ \\ \hline
$80$ & \scalebox{0.8}{$\varphi\begin{bsmallmatrix}
2 & 1 \\
2 & 2
\end{bsmallmatrix}$} & $80$ & $77$ & $80$ & $78$ & $77$ & $79$ & $78$ & $79$ & $2$ & $5$ & $2$ & $3$ & $5$ & $4$ & $3$ & $4$ \\
$78$ & \scalebox{0.8}{$\varphi\begin{bsmallmatrix}
1 & 2 \\
2 & 2
\end{bsmallmatrix}$} & $78$ & $80$ & $77$ & $80$ & $79$ & $77$ & $79$ & $78$ & $3$ & $2$ & $5$ & $2$ & $4$ & $5$ & $4$ & $3$ \\
$79$ & \scalebox{0.8}{$\varphi\begin{bsmallmatrix}
2 & 2 \\
1 & 2
\end{bsmallmatrix}$} & $79$ & $78$ & $79$ & $77$ & $78$ & $80$ & $77$ & $80$ & $4$ & $3$ & $4$ & $5$ & $3$ & $2$ & $5$ & $2$ \\
$77$ & \scalebox{0.8}{$\varphi\begin{bsmallmatrix}
2 & 2 \\
2 & 1
\end{bsmallmatrix}$} & $77$ & $79$ & $78$ & $79$ & $80$ & $78$ & $80$ & $77$ & $5$ & $4$ & $3$ & $4$ & $2$ & $3$ & $2$ & $5$ \\
$5$ & \scalebox{0.8}{$\varphi\begin{bsmallmatrix}
0 & 0 \\
0 & 1
\end{bsmallmatrix}$} & $5$ & $4$ & $3$ & $4$ & $2$ & $3$ & $2$ & $5$ & $77$ & $79$ & $78$ & $79$ & $80$ & $78$ & $80$ & $77$ \\
$3$ & \scalebox{0.8}{$\varphi\begin{bsmallmatrix}
1 & 0 \\
0 & 0
\end{bsmallmatrix}$} & $3$ & $2$ & $5$ & $2$ & $4$ & $5$ & $4$ & $3$ & $78$ & $80$ & $77$ & $80$ & $79$ & $77$ & $79$ & $78$ \\
$4$ & \scalebox{0.8}{$\varphi\begin{bsmallmatrix}
0 & 0 \\
1 & 0
\end{bsmallmatrix}$} & $4$ & $3$ & $4$ & $5$ & $3$ & $2$ & $5$ & $2$ & $79$ & $78$ & $79$ & $77$ & $78$ & $80$ & $77$ & $80$ \\
$2$ & \scalebox{0.8}{$\varphi\begin{bsmallmatrix}
0 & 1 \\
0 & 0
\end{bsmallmatrix}$} & $2$ & $5$ & $2$ & $3$ & $5$ & $4$ & $3$ & $4$ & $80$ & $77$ & $80$ & $78$ & $77$ & $79$ & $78$ & $79$ \\ \hline
$76$ & \scalebox{0.8}{$\varphi\begin{bsmallmatrix}
1 & 1 \\
2 & 2
\end{bsmallmatrix}$} & $76$ & $74$ & $74$ & $76$ & $67$ & $67$ & $69$ & $69$ & $6$ & $14$ & $14$ & $6$ & $12$ & $12$ & $8$ & $8$ \\
$69$ & \scalebox{0.8}{$\varphi\begin{bsmallmatrix}
1 & 2 \\
1 & 2
\end{bsmallmatrix}$} & $69$ & $76$ & $67$ & $74$ & $69$ & $74$ & $67$ & $76$ & $8$ & $6$ & $12$ & $14$ & $8$ & $14$ & $12$ & $6$ \\
$67$ & \scalebox{0.8}{$\varphi\begin{bsmallmatrix}
2 & 2 \\
1 & 1
\end{bsmallmatrix}$} & $67$ & $69$ & $69$ & $67$ & $76$ & $76$ & $74$ & $74$ & $12$ & $8$ & $8$ & $12$ & $6$ & $6$ & $14$ & $14$ \\
$74$ & \scalebox{0.8}{$\varphi\begin{bsmallmatrix}
2 & 1 \\
2 & 1
\end{bsmallmatrix}$} & $74$ & $67$ & $76$ & $69$ & $74$ & $69$ & $76$ & $67$ & $14$ & $12$ & $6$ & $8$ & $14$ & $8$ & $6$ & $12$ \\
$12$ & \scalebox{0.8}{$\varphi\begin{bsmallmatrix}
0 & 0 \\
1 & 1
\end{bsmallmatrix}$} & $12$ & $8$ & $8$ & $12$ & $6$ & $6$ & $14$ & $14$ & $67$ & $69$ & $69$ & $67$ & $76$ & $76$ & $74$ & $74$ \\
$8$ & \scalebox{0.8}{$\varphi\begin{bsmallmatrix}
1 & 0 \\
1 & 0
\end{bsmallmatrix}$} & $8$ & $6$ & $12$ & $14$ & $8$ & $14$ & $12$ & $6$ & $69$ & $76$ & $67$ & $74$ & $69$ & $74$ & $67$ & $76$ \\
$14$ & \scalebox{0.8}{$\varphi\begin{bsmallmatrix}
0 & 1 \\
0 & 1
\end{bsmallmatrix}$} & $14$ & $12$ & $6$ & $8$ & $14$ & $8$ & $6$ & $12$ & $74$ & $67$ & $76$ & $69$ & $74$ & $69$ & $76$ & $67$ \\
$6$ & \scalebox{0.8}{$\varphi\begin{bsmallmatrix}
1 & 1 \\
0 & 0
\end{bsmallmatrix}$} & $6$ & $14$ & $14$ & $6$ & $12$ & $12$ & $8$ & $8$ & $76$ & $74$ & $74$ & $76$ & $67$ & $67$ & $69$ & $69$ \\ \hline
$75$ & \scalebox{0.8}{$\varphi\begin{bsmallmatrix}
2 & 2 \\
2 & 0
\end{bsmallmatrix}$} & $75$ & $68$ & $70$ & $68$ & $73$ & $70$ & $73$ & $75$ & $7$ & $15$ & $9$ & $15$ & $11$ & $9$ & $11$ & $7$ \\
$70$ & \scalebox{0.8}{$\varphi\begin{bsmallmatrix}
0 & 2 \\
2 & 2
\end{bsmallmatrix}$} & $70$ & $73$ & $75$ & $73$ & $68$ & $75$ & $68$ & $70$ & $9$ & $11$ & $7$ & $11$ & $15$ & $7$ & $15$ & $9$ \\
$73$ & \scalebox{0.8}{$\varphi\begin{bsmallmatrix}
2 & 0 \\
2 & 2
\end{bsmallmatrix}$} & $73$ & $75$ & $73$ & $70$ & $75$ & $68$ & $70$ & $68$ & $11$ & $7$ & $11$ & $9$ & $7$ & $15$ & $9$ & $15$ \\
$68$ & \scalebox{0.8}{$\varphi\begin{bsmallmatrix}
2 & 2 \\
0 & 2
\end{bsmallmatrix}$} & $68$ & $70$ & $68$ & $75$ & $70$ & $73$ & $75$ & $73$ & $15$ & $9$ & $15$ & $7$ & $9$ & $11$ & $7$ & $11$ \\
$15$ & \scalebox{0.8}{$\varphi\begin{bsmallmatrix}
0 & 0 \\
2 & 0
\end{bsmallmatrix}$} & $15$ & $9$ & $15$ & $7$ & $9$ & $11$ & $7$ & $11$ & $68$ & $70$ & $68$ & $75$ & $70$ & $73$ & $75$ & $73$ \\
$9$ & \scalebox{0.8}{$\varphi\begin{bsmallmatrix}
2 & 0 \\
0 & 0
\end{bsmallmatrix}$} & $9$ & $11$ & $7$ & $11$ & $15$ & $7$ & $15$ & $9$ & $70$ & $73$ & $75$ & $73$ & $68$ & $75$ & $68$ & $70$ \\
$11$ & \scalebox{0.8}{$\varphi\begin{bsmallmatrix}
0 & 2 \\
0 & 0
\end{bsmallmatrix}$} & $11$ & $7$ & $11$ & $9$ & $7$ & $15$ & $9$ & $15$ & $73$ & $75$ & $73$ & $70$ & $75$ & $68$ & $70$ & $68$ \\
$7$ & \scalebox{0.8}{$\varphi\begin{bsmallmatrix}
0 & 0 \\
0 & 2
\end{bsmallmatrix}$} & $7$ & $15$ & $9$ & $15$ & $11$ & $9$ & $11$ & $7$ & $75$ & $68$ & $70$ & $68$ & $73$ & $70$ & $73$ & $75$ \\ \hline
$71$ & \scalebox{0.8}{$\varphi\begin{bsmallmatrix}
2 & 1 \\
1 & 2
\end{bsmallmatrix}$} & $71$ & $72$ & $71$ & $72$ & $72$ & $71$ & $72$ & $71$ & $10$ & $13$ & $10$ & $13$ & $13$ & $10$ & $13$ & $10$ \\
$72$ & \scalebox{0.8}{$\varphi\begin{bsmallmatrix}
1 & 2 \\
2 & 1
\end{bsmallmatrix}$} & $72$ & $71$ & $72$ & $71$ & $71$ & $72$ & $71$ & $72$ & $13$ & $10$ & $13$ & $10$ & $10$ & $13$ & $10$ & $13$ \\
$10$ & \scalebox{0.8}{$\varphi\begin{bsmallmatrix}
0 & 1 \\
1 & 0
\end{bsmallmatrix}$} & $10$ & $13$ & $10$ & $13$ & $13$ & $10$ & $13$ & $10$ & $71$ & $72$ & $71$ & $72$ & $72$ & $71$ & $72$ & $71$ \\
$13$ & \scalebox{0.8}{$\varphi\begin{bsmallmatrix}
1 & 0 \\
0 & 1
\end{bsmallmatrix}$} & $13$ & $10$ & $13$ & $10$ & $10$ & $13$ & $10$ & $13$ & $72$ & $71$ & $72$ & $71$ & $71$ & $72$ & $71$ & $72$ \\ \hline
$65$ & \scalebox{0.8}{$\varphi\begin{bsmallmatrix}
2 & 0 \\
1 & 2
\end{bsmallmatrix}$} & $65$ & $55$ & $65$ & $63$ & $55$ & $56$ & $63$ & $56$ & $16$ & $21$ & $16$ & $28$ & $21$ & $23$ & $28$ & $23$ \\
$55$ & \scalebox{0.8}{$\varphi\begin{bsmallmatrix}
1 & 2 \\
2 & 0
\end{bsmallmatrix}$} & $55$ & $56$ & $63$ & $56$ & $65$ & $63$ & $65$ & $55$ & $21$ & $23$ & $28$ & $23$ & $16$ & $28$ & $16$ & $21$ \\
$56$ & \scalebox{0.8}{$\varphi\begin{bsmallmatrix}
2 & 1 \\
0 & 2
\end{bsmallmatrix}$} & $56$ & $63$ & $56$ & $55$ & $63$ & $65$ & $55$ & $65$ & $23$ & $28$ & $23$ & $21$ & $28$ & $16$ & $21$ & $16$ \\
$63$ & \scalebox{0.8}{$\varphi\begin{bsmallmatrix}
0 & 2 \\
2 & 1
\end{bsmallmatrix}$} & $63$ & $65$ & $55$ & $65$ & $56$ & $55$ & $56$ & $63$ & $28$ & $16$ & $21$ & $16$ & $23$ & $21$ & $23$ & $28$ \\
$21$ & \scalebox{0.8}{$\varphi\begin{bsmallmatrix}
1 & 0 \\
0 & 2
\end{bsmallmatrix}$} & $21$ & $23$ & $28$ & $23$ & $16$ & $28$ & $16$ & $21$ & $55$ & $56$ & $63$ & $56$ & $65$ & $63$ & $65$ & $55$ \\
$23$ & \scalebox{0.8}{$\varphi\begin{bsmallmatrix}
0 & 1 \\
2 & 0
\end{bsmallmatrix}$} & $23$ & $28$ & $23$ & $21$ & $28$ & $16$ & $21$ & $16$ & $56$ & $63$ & $56$ & $55$ & $63$ & $65$ & $55$ & $65$ \\
$28$ & \scalebox{0.8}{$\varphi\begin{bsmallmatrix}
2 & 0 \\
0 & 1
\end{bsmallmatrix}$} & $28$ & $16$ & $21$ & $16$ & $23$ & $21$ & $23$ & $28$ & $63$ & $65$ & $55$ & $65$ & $56$ & $55$ & $56$ & $63$ \\
$16$ & \scalebox{0.8}{$\varphi\begin{bsmallmatrix}
0 & 2 \\
1 & 0
\end{bsmallmatrix}$} & $16$ & $21$ & $16$ & $28$ & $21$ & $23$ & $28$ & $23$ & $65$ & $55$ & $65$ & $63$ & $55$ & $56$ & $63$ & $56$ \\ \hline
$53$ & \scalebox{0.8}{$\varphi\begin{bsmallmatrix}
2 & 1 \\
2 & 0
\end{bsmallmatrix}$} & $53$ & $57$ & $58$ & $61$ & $66$ & $60$ & $59$ & $62$ & $17$ & $18$ & $25$ & $27$ & $26$ & $22$ & $19$ & $30$ \\
$57$ & \scalebox{0.8}{$\varphi\begin{bsmallmatrix}
2 & 2 \\
0 & 1
\end{bsmallmatrix}$} & $57$ & $60$ & $61$ & $62$ & $58$ & $59$ & $53$ & $66$ & $18$ & $22$ & $27$ & $30$ & $25$ & $19$ & $17$ & $26$ \\
$59$ & \scalebox{0.8}{$\varphi\begin{bsmallmatrix}
1 & 0 \\
2 & 2
\end{bsmallmatrix}$} & $59$ & $53$ & $66$ & $58$ & $62$ & $57$ & $60$ & $61$ & $19$ & $17$ & $26$ & $25$ & $30$ & $18$ & $22$ & $27$ \\
$60$ & \scalebox{0.8}{$\varphi\begin{bsmallmatrix}
0 & 2 \\
1 & 2
\end{bsmallmatrix}$} & $60$ & $59$ & $62$ & $66$ & $61$ & $53$ & $57$ & $58$ & $22$ & $19$ & $30$ & $26$ & $27$ & $17$ & $18$ & $25$ \\
$58$ & \scalebox{0.8}{$\varphi\begin{bsmallmatrix}
0 & 1 \\
2 & 2
\end{bsmallmatrix}$} & $58$ & $66$ & $53$ & $59$ & $57$ & $62$ & $61$ & $60$ & $25$ & $26$ & $17$ & $19$ & $18$ & $30$ & $27$ & $22$ \\
$66$ & \scalebox{0.8}{$\varphi\begin{bsmallmatrix}
2 & 0 \\
2 & 1
\end{bsmallmatrix}$} & $66$ & $62$ & $59$ & $60$ & $53$ & $61$ & $58$ & $57$ & $26$ & $30$ & $19$ & $22$ & $17$ & $27$ & $25$ & $18$ \\
$61$ & \scalebox{0.8}{$\varphi\begin{bsmallmatrix}
1 & 2 \\
0 & 2
\end{bsmallmatrix}$} & $61$ & $58$ & $57$ & $53$ & $60$ & $66$ & $62$ & $59$ & $27$ & $25$ & $18$ & $17$ & $22$ & $26$ & $30$ & $19$ \\
$62$ & \scalebox{0.8}{$\varphi\begin{bsmallmatrix}
2 & 2 \\
1 & 0
\end{bsmallmatrix}$} & $62$ & $61$ & $60$ & $57$ & $59$ & $58$ & $66$ & $53$ & $30$ & $27$ & $22$ & $18$ & $19$ & $25$ & $26$ & $17$ \\
$17$ & \scalebox{0.8}{$\varphi\begin{bsmallmatrix}
0 & 1 \\
0 & 2
\end{bsmallmatrix}$} & $17$ & $18$ & $25$ & $27$ & $26$ & $22$ & $19$ & $30$ & $53$ & $57$ & $58$ & $61$ & $66$ & $60$ & $59$ & $62$ \\
$18$ & \scalebox{0.8}{$\varphi\begin{bsmallmatrix}
0 & 0 \\
2 & 1
\end{bsmallmatrix}$} & $18$ & $22$ & $27$ & $30$ & $25$ & $19$ & $17$ & $26$ & $57$ & $60$ & $61$ & $62$ & $58$ & $59$ & $53$ & $66$ \\
$25$ & \scalebox{0.8}{$\varphi\begin{bsmallmatrix}
2 & 1 \\
0 & 0
\end{bsmallmatrix}$} & $25$ & $26$ & $17$ & $19$ & $18$ & $30$ & $27$ & $22$ & $58$ & $66$ & $53$ & $59$ & $57$ & $62$ & $61$ & $60$ \\
$19$ & \scalebox{0.8}{$\varphi\begin{bsmallmatrix}
1 & 2 \\
0 & 0
\end{bsmallmatrix}$} & $19$ & $17$ & $26$ & $25$ & $30$ & $18$ & $22$ & $27$ & $59$ & $53$ & $66$ & $58$ & $62$ & $57$ & $60$ & $61$ \\
$22$ & \scalebox{0.8}{$\varphi\begin{bsmallmatrix}
2 & 0 \\
1 & 0
\end{bsmallmatrix}$} & $22$ & $19$ & $30$ & $26$ & $27$ & $17$ & $18$ & $25$ & $60$ & $59$ & $62$ & $66$ & $61$ & $53$ & $57$ & $58$ \\
$27$ & \scalebox{0.8}{$\varphi\begin{bsmallmatrix}
1 & 0 \\
2 & 0
\end{bsmallmatrix}$} & $27$ & $25$ & $18$ & $17$ & $22$ & $26$ & $30$ & $19$ & $61$ & $58$ & $57$ & $53$ & $60$ & $66$ & $62$ & $59$ \\
$30$ & \scalebox{0.8}{$\varphi\begin{bsmallmatrix}
0 & 0 \\
1 & 2
\end{bsmallmatrix}$} & $30$ & $27$ & $22$ & $18$ & $19$ & $25$ & $26$ & $17$ & $62$ & $61$ & $60$ & $57$ & $59$ & $58$ & $66$ & $53$ \\
$26$ & \scalebox{0.8}{$\varphi\begin{bsmallmatrix}
0 & 2 \\
0 & 1
\end{bsmallmatrix}$} & $26$ & $30$ & $19$ & $22$ & $17$ & $27$ & $25$ & $18$ & $66$ & $62$ & $59$ & $60$ & $53$ & $61$ & $58$ & $57$ \\ \hline
$52$ & \scalebox{0.8}{$\varphi\begin{bsmallmatrix}
1 & 2 \\
1 & 1
\end{bsmallmatrix}$} & $52$ & $51$ & $52$ & $64$ & $51$ & $54$ & $64$ & $54$ & $20$ & $31$ & $20$ & $24$ & $31$ & $29$ & $24$ & $29$ \\
$64$ & \scalebox{0.8}{$\varphi\begin{bsmallmatrix}
2 & 1 \\
1 & 1
\end{bsmallmatrix}$} & $64$ & $52$ & $51$ & $52$ & $54$ & $51$ & $54$ & $64$ & $24$ & $20$ & $31$ & $20$ & $29$ & $31$ & $29$ & $24$ \\
$54$ & \scalebox{0.8}{$\varphi\begin{bsmallmatrix}
1 & 1 \\
2 & 1
\end{bsmallmatrix}$} & $54$ & $64$ & $54$ & $51$ & $64$ & $52$ & $51$ & $52$ & $29$ & $24$ & $29$ & $31$ & $24$ & $20$ & $31$ & $20$ \\
$51$ & \scalebox{0.8}{$\varphi\begin{bsmallmatrix}
1 & 1 \\
1 & 2
\end{bsmallmatrix}$} & $51$ & $54$ & $64$ & $54$ & $52$ & $64$ & $52$ & $51$ & $31$ & $29$ & $24$ & $29$ & $20$ & $24$ & $20$ & $31$ \\
$31$ & \scalebox{0.8}{$\varphi\begin{bsmallmatrix}
1 & 1 \\
1 & 0
\end{bsmallmatrix}$} & $31$ & $29$ & $24$ & $29$ & $20$ & $24$ & $20$ & $31$ & $51$ & $54$ & $64$ & $54$ & $52$ & $64$ & $52$ & $51$ \\
$20$ & \scalebox{0.8}{$\varphi\begin{bsmallmatrix}
1 & 0 \\
1 & 1
\end{bsmallmatrix}$} & $20$ & $31$ & $20$ & $24$ & $31$ & $29$ & $24$ & $29$ & $52$ & $51$ & $52$ & $64$ & $51$ & $54$ & $64$ & $54$ \\
$29$ & \scalebox{0.8}{$\varphi\begin{bsmallmatrix}
1 & 1 \\
0 & 1
\end{bsmallmatrix}$} & $29$ & $24$ & $29$ & $31$ & $24$ & $20$ & $31$ & $20$ & $54$ & $64$ & $54$ & $51$ & $64$ & $52$ & $51$ & $52$ \\
$24$ & \scalebox{0.8}{$\varphi\begin{bsmallmatrix}
0 & 1 \\
1 & 1
\end{bsmallmatrix}$} & $24$ & $20$ & $31$ & $20$ & $29$ & $31$ & $29$ & $24$ & $64$ & $52$ & $51$ & $52$ & $54$ & $51$ & $54$ & $64$ \\ \hline
$37$ & \scalebox{0.8}{$\varphi\begin{bsmallmatrix}
2 & 1 \\
1 & 0
\end{bsmallmatrix}$} & $37$ & $50$ & $32$ & $50$ & $35$ & $32$ & $35$ & $37$ & $32$ & $35$ & $37$ & $35$ & $50$ & $37$ & $50$ & $32$ \\
$50$ & \scalebox{0.8}{$\varphi\begin{bsmallmatrix}
1 & 2 \\
0 & 1
\end{bsmallmatrix}$} & $50$ & $32$ & $50$ & $37$ & $32$ & $35$ & $37$ & $35$ & $35$ & $37$ & $35$ & $32$ & $37$ & $50$ & $32$ & $50$ \\
$32$ & \scalebox{0.8}{$\varphi\begin{bsmallmatrix}
0 & 1 \\
1 & 2
\end{bsmallmatrix}$} & $32$ & $35$ & $37$ & $35$ & $50$ & $37$ & $50$ & $32$ & $37$ & $50$ & $32$ & $50$ & $35$ & $32$ & $35$ & $37$ \\
$35$ & \scalebox{0.8}{$\varphi\begin{bsmallmatrix}
1 & 0 \\
2 & 1
\end{bsmallmatrix}$} & $35$ & $37$ & $35$ & $32$ & $37$ & $50$ & $32$ & $50$ & $50$ & $32$ & $50$ & $37$ & $32$ & $35$ & $37$ & $35$ \\ \hline
$38$ & \scalebox{0.8}{$\varphi\begin{bsmallmatrix}
0 & 2 \\
0 & 2
\end{bsmallmatrix}$} & $38$ & $42$ & $47$ & $33$ & $38$ & $33$ & $47$ & $42$ & $33$ & $47$ & $42$ & $38$ & $33$ & $38$ & $42$ & $47$ \\
$33$ & \scalebox{0.8}{$\varphi\begin{bsmallmatrix}
2 & 0 \\
2 & 0
\end{bsmallmatrix}$} & $33$ & $47$ & $42$ & $38$ & $33$ & $38$ & $42$ & $47$ & $38$ & $42$ & $47$ & $33$ & $38$ & $33$ & $47$ & $42$ \\
$47$ & \scalebox{0.8}{$\varphi\begin{bsmallmatrix}
2 & 2 \\
0 & 0
\end{bsmallmatrix}$} & $47$ & $38$ & $38$ & $47$ & $42$ & $42$ & $33$ & $33$ & $42$ & $33$ & $33$ & $42$ & $47$ & $47$ & $38$ & $38$ \\
$42$ & \scalebox{0.8}{$\varphi\begin{bsmallmatrix}
0 & 0 \\
2 & 2
\end{bsmallmatrix}$} & $42$ & $33$ & $33$ & $42$ & $47$ & $47$ & $38$ & $38$ & $47$ & $38$ & $38$ & $47$ & $42$ & $42$ & $33$ & $33$ \\ \hline
$48$ & \scalebox{0.8}{$\varphi\begin{bsmallmatrix}
0 & 2 \\
1 & 1
\end{bsmallmatrix}$} & $48$ & $49$ & $43$ & $34$ & $39$ & $41$ & $45$ & $44$ & $34$ & $43$ & $49$ & $48$ & $41$ & $39$ & $44$ & $45$ \\
$41$ & \scalebox{0.8}{$\varphi\begin{bsmallmatrix}
1 & 1 \\
2 & 0
\end{bsmallmatrix}$} & $41$ & $45$ & $44$ & $39$ & $34$ & $48$ & $49$ & $43$ & $39$ & $44$ & $45$ & $41$ & $48$ & $34$ & $43$ & $49$ \\
$39$ & \scalebox{0.8}{$\varphi\begin{bsmallmatrix}
1 & 1 \\
0 & 2
\end{bsmallmatrix}$} & $39$ & $44$ & $45$ & $41$ & $48$ & $34$ & $43$ & $49$ & $41$ & $45$ & $44$ & $39$ & $34$ & $48$ & $49$ & $43$ \\
$49$ & \scalebox{0.8}{$\varphi\begin{bsmallmatrix}
1 & 0 \\
1 & 2
\end{bsmallmatrix}$} & $49$ & $41$ & $34$ & $44$ & $43$ & $45$ & $48$ & $39$ & $43$ & $39$ & $48$ & $45$ & $49$ & $44$ & $34$ & $41$ \\
$45$ & \scalebox{0.8}{$\varphi\begin{bsmallmatrix}
2 & 1 \\
0 & 1
\end{bsmallmatrix}$} & $45$ & $48$ & $39$ & $43$ & $44$ & $49$ & $41$ & $34$ & $44$ & $34$ & $41$ & $49$ & $45$ & $43$ & $39$ & $48$ \\
$44$ & \scalebox{0.8}{$\varphi\begin{bsmallmatrix}
0 & 1 \\
2 & 1
\end{bsmallmatrix}$} & $44$ & $34$ & $41$ & $49$ & $45$ & $43$ & $39$ & $48$ & $45$ & $48$ & $39$ & $43$ & $44$ & $49$ & $41$ & $34$ \\
$34$ & \scalebox{0.8}{$\varphi\begin{bsmallmatrix}
2 & 0 \\
1 & 1
\end{bsmallmatrix}$} & $34$ & $43$ & $49$ & $48$ & $41$ & $39$ & $44$ & $45$ & $48$ & $49$ & $43$ & $34$ & $39$ & $41$ & $45$ & $44$ \\
$43$ & \scalebox{0.8}{$\varphi\begin{bsmallmatrix}
1 & 2 \\
1 & 0
\end{bsmallmatrix}$} & $43$ & $39$ & $48$ & $45$ & $49$ & $44$ & $34$ & $41$ & $49$ & $41$ & $34$ & $44$ & $43$ & $45$ & $48$ & $39$ \\ \hline
$36$ & \scalebox{0.8}{$\varphi\begin{bsmallmatrix}
1 & 1 \\
1 & 1
\end{bsmallmatrix}$} & $36$ & $36$ & $36$ & $36$ & $36$ & $36$ & $36$ & $36$ & $36$ & $36$ & $36$ & $36$ & $36$ & $36$ & $36$ & $36$ \\ \hline
$46$ & \scalebox{0.8}{$\varphi\begin{bsmallmatrix}
0 & 2 \\
2 & 0
\end{bsmallmatrix}$} & $46$ & $40$ & $46$ & $40$ & $40$ & $46$ & $40$ & $46$ & $40$ & $46$ & $40$ & $46$ & $46$ & $40$ & $46$ & $40$ \\
$40$ & \scalebox{0.8}{$\varphi\begin{bsmallmatrix}
2 & 0 \\
0 & 2
\end{bsmallmatrix}$} & $40$ & $46$ & $40$ & $46$ & $46$ & $40$ & $46$ & $40$ & $46$ & $40$ & $46$ & $40$ & $40$ & $46$ & $40$ & $46$ \\
\end{tabular}
}}\vspace{.5cm}
\caption{\label{tb:3orbs} Orbits from the action of $D_4^3 = D_4\mtimes K_3$ on ternary quads; See translation of elements in Table~\ref{tb:3let2el}.}
\vspace{-2.5cm}
\end{table}

\clearpage
\subsection*{Abelian Fourier transform from the group character table.}
\label{sec:fftfromct}
There is an ordering of group elements that would simplify the construction of the Fourier transform for an arbitrary Abelian group. With this ordering, it is straightforward to read the Fourier transform from the character table for the following reasons. Because group $A$ is Abelian, all its representations are unidimensional. Therefore, all representations coincide with their characters. So to perform an Fourier transform on group $A = C_2^{\mtimes 4}$ using Equation~\ref{eq:fft} we can use the group characters for the irreducible representations. The whole transform is then just a matter of multiplying the bases in Equation~\ref{eq:quads} by the character table (as a matrix) of the group $A$ to create the Fourier transform.

The ordering above of group elements, or equivalently the perception space coordinates, makes explicit the fact that this multidimensional commutative Fourier transform is actually a special case of a Fourier transform known as the Walsh-Hadamard-Rademacher transform~\cite{kunzEquivalenceOneDimensionalDiscrete1979}, implemented by a natural order Hadamard matrix as $\varphi = H(16)p$. Another ordering of this same transform yields what is called a sequence ordered Hadamard matrix or in some corners a Walsh Matrix. Yet other orderings would change the pattern on this matrix and highlight some other fundamental property of this transform. We show that the order in~\cite{victorLocalImageStatistics2012}, can be obtained from an alternative derivation by making use of the character table for the Abelian group $A = C_2^{4\mtimes4}$, see {\bf SI}. The character table list elements of $G$ by equivalent classes in columns and irreps in rows.

\begin{table}[htbp]
\centering
{\fontfamily{cmtt}\selectfont\footnotesize
\begin{tabular}{r|rrrrrrrrrrrrrrrr}
& 1 & 2 & 3 & 4 & 5 & 6 & 7 & 8 & 9 & 10 & 11 & 12 & 13 & 14 & 15 & 16 \\
\hline
 1 & 1 & 1 & 1 & 1 & 1 & 1 & 1 & 1 & 1 & 1 & 1 & 1 & 1 & 1 & 1 & 1 \\
2 & 1 & -1 & -1 & 1 & -1 & 1 & 1 & -1 & -1 & 1 & 1 & -1 & 1 & -1 & -1 & 1 \\
3 & 1 & -1 & -1 & 1 & -1 & 1 & 1 & -1 & 1 & -1 & -1 & 1 & -1 & 1 & 1 & -1 \\
4 & 1 & -1 & -1 & 1 & 1 & -1 & -1 & 1 & -1 & 1 & 1 & -1 & -1 & 1 & 1 & -1 \\
5 & 1 & -1 & -1 & 1 & 1 & -1 & -1 & 1 & 1 & -1 & -1 & 1 & 1 & -1 & -1 & 1 \\
6 & 1 & -1 & 1 & -1 & -1 & 1 & -1 & 1 & -1 & 1 & -1 & 1 & 1 & -1 & 1 & -1 \\
7 & 1 & -1 & 1 & -1 & -1 & 1 & -1 & 1 & 1 & -1 & 1 & -1 & -1 & 1 & -1 & 1 \\
8 & 1 & -1 & 1 & -1 & 1 & -1 & 1 & -1 & -1 & 1 & -1 & 1 & -1 & 1 & -1 & 1 \\
9 & 1 & -1 & 1 & -1 & 1 & -1 & 1 & -1 & 1 & -1 & 1 & -1 & 1 & -1 & 1 & -1 \\
10 & 1 & 1 & -1 & -1 & -1 & -1 & 1 & 1 & -1 & -1 & 1 & 1 & 1 & 1 & -1 & -1 \\
11 & 1 & 1 & -1 & -1 & -1 & -1 & 1 & 1 & 1 & 1 & -1 & -1 & -1 & -1 & 1 & 1 \\
12 & 1 & 1 & -1 & -1 & 1 & 1 & -1 & -1 & -1 & -1 & 1 & 1 & -1 & -1 & 1 & 1 \\
13 & 1 & 1 & -1 & -1 & 1 & 1 & -1 & -1 & 1 & 1 & -1 & -1 & 1 & 1 & -1 & -1 \\
14 & 1 & 1 & 1 & 1 & -1 & -1 & -1 & -1 & -1 & -1 & -1 & -1 & 1 & 1 & 1 & 1 \\
15 & 1 & 1 & 1 & 1 & -1 & -1 & -1 & -1 & 1 & 1 & 1 & 1 & -1 & -1 & -1 & -1 \\
16 & 1 & 1 & 1 & 1 & 1 & 1 & 1 & 1 & -1 & -1 & -1 & -1 & -1 & -1 & -1 & -1
\end{tabular}
}
\caption{Character Table of $A = C_2^{\mtimes 4}$ can be used to calculate the Fourier transform on $A$ via matrix multiplication. There is a particular ordering of the group elements by conjugacy classes representatives in the columns of the table. The rows follow a particular ordering of the irreducible representations.} \label{tb:charactertable}
\end{table}
\begin{table}[htbp] 
\centering
{\fontfamily{cmtt} \selectfont \footnotesize 
\begin{tabular}{r|rrrrrrrrrrrrrrrr}
 & 1 & 9 & 5 & 13 & 3 & 11 & 7 & 15 & 2 & 10 & 6 & 14 & 4 & 12 & 8 & 16 \\
\hline
 1 & 1 & 1 & 1 & 1 & 1 & 1 & 1 & 1 & 1 & 1 & 1 & 1 & 1 & 1 & 1 & 1 \\
16 & 1 & -1 & 1 & -1 & 1 & -1 & 1 & -1 & 1 & -1 & 1 & -1 & 1 & -1 & 1 & -1 \\
15 & 1 & 1 & -1 & -1 & 1 & 1 & -1 & -1 & 1 & 1 & -1 & -1 & 1 & 1 & -1 & -1 \\
14 & 1 & -1 & -1 & 1 & 1 & -1 & -1 & 1 & 1 & -1 & -1 & 1 & 1 & -1 & -1 & 1 \\
13 & 1 & 1 & 1 & 1 & -1 & -1 & -1 & -1 & 1 & 1 & 1 & 1 & -1 & -1 & -1 & -1 \\
12 & 1 & -1 & 1 & -1 & -1 & 1 & -1 & 1 & 1 & -1 & 1 & -1 & -1 & 1 & -1 & 1 \\
11 & 1 & 1 & -1 & -1 & -1 & -1 & 1 & 1 & 1 & 1 & -1 & -1 & -1 & -1 & 1 & 1 \\
10 & 1 & -1 & -1 & 1 & -1 & 1 & 1 & -1 & 1 & -1 & -1 & 1 & -1 & 1 & 1 & -1 \\
9 & 1 & 1 & 1 & 1 & 1 & 1 & 1 & 1 & -1 & -1 & -1 & -1 & -1 & -1 & -1 & -1 \\
8 & 1 & -1 & 1 & -1 & 1 & -1 & 1 & -1 & -1 & 1 & -1 & 1 & -1 & 1 & -1 & 1 \\
7 & 1 & 1 & -1 & -1 & 1 & 1 & -1 & -1 & -1 & -1 & 1 & 1 & -1 & -1 & 1 & 1 \\
6 & 1 & -1 & -1 & 1 & 1 & -1 & -1 & 1 & -1 & 1 & 1 & -1 & -1 & 1 & 1 & -1 \\
5 & 1 & 1 & 1 & 1 & -1 & -1 & -1 & -1 & -1 & -1 & -1 & -1 & 1 & 1 & 1 & 1 \\
4 & 1 & -1 & 1 & -1 & -1 & 1 & -1 & 1 & -1 & 1 & -1 & 1 & 1 & -1 & 1 & -1 \\
3 & 1 & 1 & -1 & -1 & -1 & -1 & 1 & 1 & -1 & -1 & 1 & 1 & 1 & 1 & -1 & -1 \\
2 & 1 & -1 & -1 & 1 & -1 & 1 & 1 & -1 & -1 & 1 & 1 & -1 & 1 & -1 & -1 & 1
\end{tabular}}
\caption{Reordered character table of $A = C_2^{\mtimes 4}$. The Fourier Transform (Equation~\ref{eq:fft}), is equivalent to a product of a matrix formed by the entries of this table and a vector of patch probabilities. }
\label{tb:reorderedct}
\end{table}

\small
\clearpage
\bibliographystyle{unsrt}  
\bibliography{dihedral.bib}

\end{document}